\documentclass[%
 reprint,
 twocolumn,
 amsmath,amssymb,
 aps,
 prl,
 superscriptaddress,
]{revtex4-2}
\usepackage{graphicx}
\usepackage{bm}
\usepackage{subcaption}
\graphicspath{{figures/}{images/}{plots/}}
\usepackage{amsmath}
\usepackage{amssymb}
\usepackage{amsthm}
\usepackage{graphicx}
\usepackage{mathrsfs}
\usepackage{braket}
\usepackage{comment}
\usepackage[hidelinks]{hyperref}
\usepackage[capitalize]{cleveref}
\usepackage[dvipsnames]{xcolor}
\usepackage{soul,xcolor}
\usepackage{orcidlink}
\usepackage{color, xcolor, colortbl}
\usepackage{graphicx}
\usepackage{amsthm,amsmath,amssymb,amsfonts}
\usepackage{dcolumn}
\usepackage{url}
\usepackage{enumitem}
\usepackage{algpseudocode}
\usepackage{bm}
\usepackage{appendix}
\usepackage{multirow}
\usepackage{braket}
\usepackage[english]{babel}
\usepackage{aliascnt}
\usepackage{hyperref}
\usepackage[capitalize]{cleveref}
\usepackage{algorithm}
\usepackage{xpatch}
\usepackage{tikz}
\usepackage{adjustbox}
\usepackage{xspace}
\usetikzlibrary{quantikz}
\usepackage{pifont}

\newcommand{\Tr}{\operatorname{Tr}}

\newcommand{\mc}[1]{\mathcal{#1}}

\newtheorem{thm}{Theorem}[section]
\crefname{thm}{Theorem}{Theorems}
\Crefname{thm}{Theorem}{Theorems}

\newaliascnt{lem}{thm}
\newaliascnt{rem}{thm}
\newaliascnt{prop}{thm}
\newaliascnt{cor}{thm}
\newaliascnt{assumption}{thm}
\newaliascnt{defn}{thm}
\newaliascnt{notation}{thm}
\newaliascnt{fact}{thm}
\newaliascnt{setup}{thm}
\newaliascnt{condition}{thm}

\newtheorem{condition}[condition]{Condition}
\newtheorem{lem}[lem]{Lemma}

\newtheorem{prop}[prop]{Proposition}

\newtheorem{assumption}[assumption]{Assumption}

\aliascntresetthe{lem}
\aliascntresetthe{rem}
\aliascntresetthe{prop}
\aliascntresetthe{cor}
\aliascntresetthe{assumption}
\aliascntresetthe{defn}
\aliascntresetthe{notation}
\aliascntresetthe{fact}
\aliascntresetthe{setup}
\aliascntresetthe{condition}

\crefname{lem}{Lemma}{Lemmas}
\Crefname{lem}{Lemma}{Lemmas}

\crefname{rem}{Remark}{Remarks}
\Crefname{rem}{Remark}{Remarks}

\crefname{prop}{Proposition}{Propositions}
\Crefname{prop}{Proposition}{Propositions}

\crefname{cor}{Corollary}{Corollaries}
\Crefname{cor}{Corollary}{Corollaries}

\crefname{assumption}{Assumption}{Assumptions}
\Crefname{assumption}{Assumption}{Assumptions}

\crefname{defn}{Definition}{Definitions}
\Crefname{defn}{Definition}{Definitions}

\crefname{notation}{Notation}{Notations}
\Crefname{notation}{Notation}{Notations}

\crefname{condition}{Condition}{Conditions}
\Crefname{condition}{Condition}{Conditions}

\crefname{fact}{Fact}{Facts}
\Crefname{fact}{Fact}{Facts}
\Crefname{setup}{Setup}{Setups}

\usepackage{bbm}

\NewDocumentCommand{\ketbra}{mG{#1}}{\mathinner{|{#1}\rangle\!\langle{#2}|}}

\usepackage{silence}


\usetikzlibrary{fit}
\tikzset{%
  highlight/.style={rectangle,rounded corners,fill=blue!15,draw,fill opacity=0.3,thick,inner sep=0pt}
}

\counterwithout{thm}{section}
\aliascntresetthe{lem}
\aliascntresetthe{rem}
\aliascntresetthe{prop}
\aliascntresetthe{cor}
\aliascntresetthe{assumption}
\aliascntresetthe{defn}
\aliascntresetthe{notation}
\aliascntresetthe{fact}
\aliascntresetthe{setup}
\aliascntresetthe{condition}

\graphicspath{{figures/}}

\newcommand{\DeptMath}{Department of Mathematics, University of Michigan, Ann Arbor, MI 48109, USA}
\begin{document}

\title{Sublinear thermal state and ground state preparation for free fermions via system-bath interaction}

\author{Zhiyan Ding\orcidlink{0000-0001-8863-403X}}
\email{zyding@umich.edu}
\affiliation{\DeptMath}
\author{Lin Lin\orcidlink{0000-0001-6860-9566}}
\email{lin@caltech.edu}
\affiliation{Department of Computing and Mathematical Sciences, California Institute of Technology, Pasadena, CA 91125, USA}
\affiliation{Applied Mathematics and Computational Research Division, Lawrence Berkeley National Laboratory, Berkeley, CA 94720, USA}
\affiliation{Department of Mathematics, University of California, Berkeley, CA 94720, USA}
\author{Ke Wang\orcidlink{0009-0001-3225-887X}}
\email{kwmath@umich.edu}
\affiliation{\DeptMath}


\begin{abstract}

Dissipative algorithms for preparing thermal and ground states typically require a total Hamiltonian simulation time that grows at least linearly with system size $N$, even when their mixing time grows only polylogarithmically. We achieve sublinear simulation time for free fermions by engineering local interactions with a bath whose size grows in proportion to the system size. Starting from arbitrary initial states, our thermal state preparation algorithm achieves trace-norm error $\epsilon$ in total simulation time $\widetilde{\mathcal{O}}(\sqrt{N}/\epsilon)$ at fixed inverse temperature and bounded single-particle bandwidth, while ground state preparation of gapped systems requires only $\mathrm{polylog}(N/\epsilon)$ time. We further establish lower bounds for thermal state preparation within the bath protocols considered, identifying two obstacles to faster preparation, namely fourth-order deviations from the target thermal state and fluctuations induced by a random frequency shared by all bath qubits. Our results demonstrate that the linear simulation time barrier in dissipative state preparation can be overcome and clarify the bath and filter-design constraints that govern further complexity improvements.
\end{abstract}

\maketitle

{\it Introduction} --
Quantum state preparation is a foundational challenge in quantum computation and simulation. Given a Hamiltonian $H$, its thermal and ground states encode key equilibrium properties of the underlying quantum system and serve as essential resources in quantum many-body physics, quantum chemistry, and materials science.

Dissipative state preparation algorithms use controlled interactions with an environment or bath to drive the system toward the desired state. Early heuristic studies provided important evidence that engineered dissipation can successfully prepare pure and thermal quantum states in specific structured settings~\cite{KrausBuchlerDiehlEtAl2008,DiehlMicheliKantianEtAl2008,VerstraeteWolfIgnacioCirac2009,KastoryanoReiterSorensen2011,KrauterMuschikJensenEtAl2011,Temme_2011,RoyChalkerGornyiEtAl2020,zhou2021symmetry,shtanko2023}. Recently, substantial progress has been made toward systematically designed dissipative algorithms with rigorous performance and complexity guarantees. One direction involves designing Lindblad generators whose stationary states coincide with, or closely approximate, the desired Gibbs or ground states, together with algorithms and analyses establishing end-to-end complexity guarantees in relevant regimes~\cite{Cubitt2023,Chen2025Efficient,ChenKastoryanoGilyen2023,Ding2024,Ding_2025,li2024dissipative,gilyen2024,hahn2025,rouze2024efficient,zhan2025rapidquantumgroundstate}. More recently, system-bath interaction algorithms have provided a complementary approach, in which the dissipative dynamics is implemented directly through interactions with a bath that is reset after each round of evolution. A growing body of work has developed such constructions and established increasingly rigorous guarantees on stationary-state accuracy and mixing behavior~\cite{hahn2026efficientquantumthermalstate,ding2025endtoendefficientquantumthermal,lloyd2025quantumthermal,slezak2026polynomialtime,wang2025lindblad,chen2026overcominglambshiftsystembath}.

In recent developments, dissipative algorithms can be designed without explicit knowledge of the target state, requiring only access to Hamiltonian simulation of the form $\exp(-iHt)$. The end-to-end complexity of such algorithms is measured by the total Hamiltonian simulation time required to prepare the target state to a prescribed accuracy. Let $N$ denote the system size. For Lindbladian-based algorithms, the state-of-the-art complexity scales as $\widetilde{\mathcal{O}}(N)$ for thermal- and ground state preparation in regimes of weak interactions~\cite{zhan2025rapidquantumgroundstate,rouze2024optimal,smid2025rapidmixingquantumgibbs}, while a complexity of $\widetilde{\mathcal{O}}(N^4)$ has been established for free-fermion ground state preparation~\cite{wang2025lindblad} with single ancilla bath.

An at least linear dependence on the system size $N$ is a common feature of existing dissipative algorithms. This dependence is closely tied to the locality of the dissipative transitions and the need to drive the entire system toward the target state. Consider Lindbladian-based algorithms as an example. The jump operators are often constructed from Heisenberg-picture evolutions of local operators and therefore remain quasi-local, acting effectively on only a constant region of the system. To reach the target state from an arbitrary initial state, dissipative transitions must collectively act on all system qubits. In typical local constructions, this requires at least $\Theta(N)$ applications of such dissipative transitions. Consequently, the total Hamiltonian simulation time naturally exhibits a lower bound of the form $T_{\rm total}=\Omega(N)$, explaining the at least linear system-size dependence that appears in existing approaches. On the other hand, the locality of the interaction is often a key feature of dissipative algorithms, as it allows for efficient implementation of the algorithm to ensure that the dynamics remains physically realizable. This raises a natural question:

\emph{Can dissipative algorithms prepare thermal and ground states with sublinear total Hamiltonian simulation time in the system size while retaining local interactions?}

In this work, we show that the linear-in-$N$ barrier can be overcome for free fermions using the recently developed system-bath interaction framework. For thermal state preparation to a target accuracy $\epsilon$, locally coupling the system to $\Theta(N)$ bath qubits reduces the total Hamiltonian simulation time to $\widetilde{\mathcal{O}}(\sqrt{N}/\epsilon)$, yielding a polynomial improvement over existing dissipative approaches. For ground state preparation, the complexity further improves to $\widetilde{\mathcal{O}}(\operatorname{polylog}(N/\epsilon))$, achieving a nearly optimal dependence on the system size. We also establish a lower bound for thermal state preparation, demonstrating the sharpness of our analysis in the regime considered and identifying the design bottlenecks that limit further improvement.

Although free-fermion systems are classically simulable, establishing efficient quantum state-preparation algorithms for this benchmark remains of significant interest. Free fermions provide a natural testbed for exploring the potential of dissipative state-preparation algorithms~\cite{molpeceres2025} and have also served as a benchmark for studying their robustness to noise~\cite{molpeceres2026benchmarkquantumalgorithmsground}. Our results provide rigorous complexity guarantees that are sublinear in the system size $N$ for this benchmark problem under the framework of dissipative state preparation.


In our construction, we restrict the total Hamiltonian to be quadratic, which is necessary for the system-bath interaction dynamics to preserve the free-fermion structure. Importantly, however, convergence of the algorithm does not require the initial state to be a free-fermion state and holds for arbitrary initial states. This suggests that the underlying mechanism is not intrinsically tied to free-fermion initial data and may admit extensions to more general quantum systems.

Employing an extensive bath, whose size grows in proportion to the system size, has become a common design choice in dissipative algorithms~\cite{hahn2026efficientquantumthermalstate,lloyd2025quantumthermal,molpeceres2025}, yet it is often motivated heuristically, and it remains theoretically unclear whether such a choice can yield a provable improvement in end-to-end complexity. We note that, while increasing the number of bath qubits can enhance the available dissipative transitions, it also introduces additional fluctuations that may drive the system away from the target state. Balancing these competing effects is nontrivial and lies beyond the scope of existing analyses of dissipative algorithms. Consequently, it is not a priori clear whether the advantages of an extensive bath outweigh its potential drawbacks. Our results rigorously demonstrate the advantages of employing an extensive bath over existing system-bath interaction algorithms and their analyses. Furthermore, they establish the use of an extensive bath as a promising design principle for overcoming the linear-in-$N$ complexity barrier in dissipative state preparation, a regime that, to the best of our knowledge, remains beyond the reach of existing Lindbladian-based algorithms.


{\it System-bath interaction algorithm} --
Given an $N$-qubit system Hamiltonian $H$ and an inverse temperature
$0<\beta<\infty$, the thermal state is defined as $\sigma_\beta\propto e^{-\beta H}$, while we denote the ground state by $\ketbra{\psi_0}$. The system-bath
interaction algorithm~\cite{hahn2026efficientquantumthermalstate,
ding2025endtoendefficientquantumthermal,lloyd2025quantumthermal,
molpeceres2025} prepares the target state by iteratively applying a
quantum channel $\Phi_\alpha$ induced by a joint system-bath evolution
followed by tracing out the bath. Specifically, each iteration is given by
\begin{equation}\label{eqn:Phi_alpha}
\rho_{n+1}
=
\Phi_{\alpha}(\rho_n)
:=
\mathbb{E}\left[
\Tr_E\left(
U^\alpha(T)
\left(\rho_n\otimes\rho_E\right)
U^\alpha(T)^\dagger
\right)
\right].
\end{equation}
Here, $\rho_E$ denotes the initial bath state at the beginning of each iteration, and $U^\alpha(t)
:=
\mathcal{T}\exp\left(
-i\int_{-T}^{t} H_\alpha(s)\,\mathrm{d}s
\right)$ is the time-evolution operator governing the joint system-bath dynamics
up to time $t$. We take the joint Hamiltonian to be
\begin{equation}\label{eqn:H_alpha_formula}
H_\alpha(t)
=
H+H_E+\alpha H_{\rm int}(t),
\end{equation}
where $H_E$ is the bath Hamiltonian, $H_{\rm int}(t)$ describes the
system-bath interaction, and $\alpha$ controls the interaction strength. The choice of $H_E$, $H_{\rm int}(t)$, and $\rho_E$ might contain randomness, and the expectation $\mathbb{E}[\cdot]$ in~\eqref{eqn:Phi_alpha} is taken over the randomness in the construction of the system-bath interaction algorithm.

In designing a system-bath interaction algorithm, the key is to choose
$H_E$, $H_{\rm int}(t)$, and $\rho_E$ such that the resulting quantum
channel $\Phi_\alpha$ admits a unique stationary state close to the target
state. In practice, the interaction Hamiltonian is often constructed as a
sum of local system-bath couplings of the form
\begin{equation}\label{eqn:H_int_formula}
H_{\rm int}(t)
=
f(t)\sum_i A_i\otimes B_i+\mathrm{h.c.},
\end{equation}
where $A_i$ and $B_i$ are local operators acting on the system and bath,
respectively, and $f(t)$ is a complex valued function. Each
term $A_i\otimes B_i$ induces transitions that exchange energy between
the system and bath, while $f(t)$, often referred to as a filter function,
is designed to control the transition rates between different energy
levels. Assuming access to the joint evolution with simultaneous system-bath couplings, we characterize the
end-to-end complexity of the algorithm by the total Hamiltonian simulation
time $T_{\rm total}=2T\,m$, where $m$ denotes the number of iterations and $2T$ is the Hamiltonian evolution time required in each iteration.

{\it Main results} --
We assume $H=\sum_{i,j} h_{ij} c^\dagger_ic_j$, where $h=(h_{ij})$ is a Hermitian matrix. Since $h$ is Hermitian, there exists a unitary matrix $U$ such that $h=U\Lambda U^\dagger$. Then,
\begin{equation}\label{eqn:H_diagonal}
\begin{aligned}
H&=\sum^N_{k=1}\lambda_k\left(\sum_j (U^\dagger)_{k,j}c_j\right)^\dagger\left(\sum_j (U^\dagger)_{k,j}c_j\right)\\
&:=\sum^N_{k=1}\lambda_k d^\dagger_k d_k\,.
\end{aligned}
\end{equation}
Because $U$ is unitary, $\{d^\dagger_k,d_k\}$ forms a new set of creation and annihilation operators.

For free fermions, we use a bath with zero Hamiltonian for thermal state preparation and a bath with a randomly sampled frequency for ground state preparation. We call these the zero-bath and random-bath settings, respectively. In both cases, the construction uses $\Theta(N)$ bath qubits, and $H_{\rm int}(t)$ in~\eqref{eqn:H_int_formula} consists of $\Theta(N)$ local system-bath couplings. Moreover, the choices of the local operators $A_i$ and $B_i$ are independent of $h$. The main results are summarized in the following theorem:
\begin{thm}[Sublinear thermal state preparation for free fermions]\label{thm:main_thermal} Let $\Lambda=\max_k|\lambda_k|=\mathcal O(1)$. There exists a system-bath interaction algorithm $\Phi_{\alpha}$ with $T=\widetilde{\Theta}(1+\beta)$, such that when $m=\widetilde{\Theta}\left(\frac{\min\{N^{1/2}\exp(\beta\Lambda/2),N\}}{\epsilon}\right)$, $\sup_{\rho}\left\|\Phi_{\alpha}^m(\rho)-\sigma_{\beta}\right\|_1\leq \epsilon$. In particular, when $\beta=\Theta(1)$, $\Lambda=\Theta(1)$, to prepare the Gibbs state $\sigma_\beta$ with $\epsilon$-accuracy, the total Hamiltonian simulation time is
\begin{equation}\label{eqn:total_time}
T_{\rm total}=\widetilde{\mathcal{O}}(N^{1/2}/\epsilon)\,.
\end{equation}
Here, $\widetilde{\Theta},\widetilde{\mathcal{O}}$ hide some polylogarithmic factors of $\beta,N,1/\epsilon$.
\end{thm}

The rigorous statement of the theorem is presented in~\cref{sec:zero_bath_thermal}~\cref{thm:thermal_zero_bath}, with the proof provided in \cref{sec:proof_thermal}. The construction of the system-bath interaction algorithm is based on the choice $H_E=0$, referred to as the zero-bath setting  (see~\cref{assump:zero_bath}), together with a suitably chosen filter function $f(t)$. The design of the filter function is motivated by recent developments in Lindbladian-based algorithms~\cite{Ding_2025}. Specifically, $f(t)$ is chosen to ensure that the leading-order dissipative generator in the infinite-window limit satisfies the Kubo–Martin–Schwinger (KMS) detailed-balance condition.

In contrast to previous system–bath interaction algorithms and analyses~\cite{ding2025endtoendefficientquantumthermal,wang2025lindblad}, our algorithm employs $\Theta(N)$ ancilla bath qubits to accelerate the mixing of the quantum channel $\Phi_{\alpha}$. However, the increased number of system–bath coupling operators introduces additional perturbation errors, which require the coupling strength $\alpha$ to remain sufficiently small. Consequently, a direct extension of previous analyses still yields a complexity that scales at least linearly with $N$.

To control these errors, we use a transformation that decomposes the total system--bath Hamiltonian into $N$ independent fermionic modes. This decoupling enables us to analyze the perturbation error of each mode separately using the fermionic tensor-product structure, thereby allowing a larger coupling strength $\alpha$ and achieving sublinear-in-$N$ complexity. To the best of our knowledge, it remains unclear whether Lindbladian-based algorithms~\cite{ChenKastoryanoGilyen2023,Ding_2025} can achieve a similar sublinear-in-$N$ complexity through such decoupling techniques, as their implementations still require $\Theta(N)$ jump operators in this setting.


Our second main result concerns ground state preparation, which is summarized as follows:
\begin{thm}[Polylogarithmic ground state preparation for free fermions]\label{thm:main_ground}
Assume $\max_k|\lambda_k|=\mathcal O(1)$ and $\Delta:=\min_k|\lambda_k|>0$. There exists a system-bath interaction algorithm $\Phi_{\alpha}$ with $T=\widetilde{\Theta}(\Delta^{-1})$, such that when $m=\widetilde{\Theta}\left(\Delta^{-1}\mathrm{polylog}(N/\epsilon)\right)$, $\sup_{\rho}\left\|\Phi_{\alpha}^m(\rho)-\ket{\psi_0}\bra{\psi_0}\right\|_1\leq \epsilon$. In particular, to prepare the ground state $\ket{\psi_0}$ with $\epsilon$-accuracy, the total Hamiltonian simulation time is $T_{\rm total}=\widetilde{\mathcal{O}}(\Delta^{-2})$.
\end{thm}

A detailed version of the above theorem is put in \cref{sec:state_random_bath}~\cref{thm:random_bath_ground} with the proof provided in \cref{sec:pf_mixing_time_random_bath_ground}. The construction of the system-bath interaction algorithm is based on a random-bath setting (see~\cref{assump:zero_bath}), where $H_E$ is chosen to be a random diagonal matrix and the filter function $f(t)$ is chosen properly.

We note that achieving the $\mathrm{polylog}(1/\epsilon)$ dependence requires a sufficiently large interaction strength $\alpha=\widetilde{\Theta}(\sqrt{\Delta})$, for which the resulting quantum channel $\Phi_\alpha$ can no longer be accurately approximated by Lindbladian evolution~\cite{wang2025lindblad}. The analysis therefore requires going beyond the conventional Lindblad approximation for open quantum systems. In particular, it is necessary to account for all higher-order terms in the Dyson expansion and establish that their contributions preserve the ground state to the desired accuracy without introducing additional system-size dependencies.

{\it Zero-bath and random-bath settings} -- In this work, we consider two classes of engineered baths: the zero-bath and random-bath settings.


\begin{assumption}[Zero bath and random bath]\label{assump:zero_bath}
For the system-bath interaction terms, we choose
\[
H_{\rm int}(t)=f(t)\sum^N_{i=1}\left(c^\dagger_i\otimes b^\dagger_{i,+}+c_i\otimes b^\dagger_{i,-}\right)+\mathrm{h.c.}\,.
\]
Here, $\{b^\dagger_{j,+},b_{j,+}\}^N_{j=1}$ and $\{b^\dagger_{j,-},b_{j,-}\}^N_{j=1}$ are two families of fermionic creation and annihilation operators on a single $2N$-mode fermionic bath and commute with system operators. In addition,
\begin{itemize}
  \item (Zero-bath) $H_E=0$ and $\rho_E=\ket{0}_+\bra{0}_+\otimes \ket{0}_-\bra{0}_-$, where $\ket{0}_+$ and $\ket{0}_-$ are the vacuum states of the two mode families.
  \item (Random-bath) $H_E=\omega\sum^{N}_{i=1}\left(b^\dagger_{i,+}b_{i,+}+b^\dagger_{i,-}b_{i,-}\right)$ and
  $
\rho_E=
\left[
(1-p_{\rm b}(\omega))|0\rangle\langle0|
+p_{\rm b}(\omega)|1\rangle\langle1|
\right]^{\otimes 2N}.
$ Here, \(\omega\) is sampled from a distribution \(\mu\), and \(p_{\rm b}(\omega)\in[0,1]\).
\end{itemize}
\end{assumption}

Although the above assumptions differ only in the choices of $H_E$ and $\rho_E$, their analyses are very different. In the random-bath setting, the additional randomness provides greater flexibility in designing the filter function $f(t)$. In particular, the choice of $f$ in~\cref{assump:zero_bath_ground} ensures that the Dyson expansion of the quantum channel $\Phi_\alpha$ preserves the target state to arbitrarily high order in the interaction strength $\alpha$. Consequently, at fixed spectral gap, we can choose a coupling strength $\alpha=\tilde{\Theta}(1)$ while maintaining arbitrarily high accuracy in preserving the target state. Moreover, the mixing time of $\Phi_\alpha$ remains $\tilde{\Theta}(1)$. Together, these properties lead to the polylogarithmic complexity established in~\cref{thm:main_ground}.
However, the same advantage does not extend to thermal state preparation: randomness shared across the bath qubits introduces additional fluctuations that obstruct sublinear complexity, as shown in~\cref{thm:random_bath_lower_bound}. By contrast, although the zero-bath setting provides less flexibility in filter design, it avoids the fluctuations induced by a shared random bath frequency. This makes it possible to achieve thermal state preparation with complexity $\widetilde{\mathcal{O}}(\sqrt{N}/\epsilon)$ in~\cref{thm:main_thermal}, yielding a polynomial improvement over existing dissipative algorithms.

{\it Lower bounds for thermal state preparation} -- As established in~\cref{thm:main_thermal,thm:main_ground}, thermal state preparation can be achieved with sublinear total Hamiltonian simulation time, although its system-size dependence remains less favorable than the polylogarithmic scaling achieved for ground state preparation. To explore the possibility of further improvements and identify the limitations of existing choices of system-bath interaction algorithms, we establish: (i). Two lower bounds for thermal state preparation in the zero-bath and random-bath settings, presented in~\cref{thm:zero_bath_lower_bound_new_1,thm:random_bath_lower_bound}, respectively. (ii). A \emph{necessary and sufficient condition} for achieving improved complexity in the zero-bath setting, presented in~\cref{thm:zero_bath_lower_bound_rigor}. For the simplicity of presentation, in this section, we assume $\beta=1$ and $\|h\|\leq 1$ and only introduce the two lower bounds result. We put~\cref{thm:zero_bath_lower_bound_rigor} and technical details in~\cref{sec:zero_bath_lower_bound}.

We start with zero-bath setting. Our first result shows that the filter functions in~\cref{thm:main_thermal} give a lower bound that matches the upper bound in~\cref{thm:main_thermal} up to polylogarithmic factors.

\begin{thm}[Lower bound result for zero-bath]\label{thm:zero_bath_lower_bound_new_1}
For $N\ge2$, consider the system-bath interaction algorithm under~\cref{assump:zero_bath}. If the filter function $f(t)$ is chosen as~\cref{thm:main_thermal} (\cref{thm:thermal_zero_bath}), then, whenever $T^{-1}$, $\alpha$, and $\epsilon$ are smaller than sufficiently small constants independent of $N$ and
$\int_{|t|>T/2}|f(t)|\mathrm{d}t=\mathcal{O}(\alpha^4)$ with a fixed implicit constant, achieving $\epsilon$-accuracy for thermal state preparation requires, in the worst case,
$T_{\rm total}=\Omega(\sqrt{N}/\epsilon)$.
\end{thm}
While~\cref{thm:zero_bath_lower_bound_new_1} establishes the sharpness of the upper bound in~\cref{thm:main_thermal} for the specific filter functions used in~\cref{thm:main_thermal}, it won't rule out the possibility of achieving improved complexity with a different choice of filter function. To complete the whole story, we establish a \emph{necessary and sufficient condition} for achieving improved complexity in the zero-bath setting for general free fermion Hamiltonians in~\cref{thm:zero_bath_lower_bound_rigor} of~\cref{sec:zero_bath_lower_bound}, thereby providing a precise characterization of the design constraints on the filter function $f(t)$. In \cref{cond:zero_bath_improvement} of ~\cref{sec:zero_bath_lower_bound}, we will show that the necessary condition for achieving improved complexity is to satisfy the detailed balance condition for higher order (fourth-order) terms in the interaction strength $\alpha$, which adds another constraint to the design of the filter function $f(t)$. Unfortunately, it remains an open question whether there exists a filter function that satisfies the improvement condition.

Next, we consider the random-bath setting, where the additional random frequency $\omega$ provides greater flexibility in designing the filter function $f(t)$. Specifically, given $f$, we define the leading order of the transition rates as follows:
\begin{equation}\label{eqn:u_d_definition}
\begin{aligned}
&u(\omega;\lambda)
=\bigl(1-p_{\rm b}(\omega)\bigr)
|\widehat f(\lambda+\omega)|^2
+p_{\rm b}(\omega)
|\widehat f(\omega-\lambda)|^2,\\
&d(\omega;\lambda)
=\bigl(1-p_{\rm b}(\omega)\bigr)
|\widehat f(\omega-\lambda)|^2
+p_{\rm b}(\omega)
|\widehat f(\lambda+\omega)|^2.
\end{aligned}
\end{equation}
Similar to the calculation in the proof of~\cref{thm:main_ground}, to make sure the fixed point is an approximately thermal state when $\alpha$ is small, we need the leading order detailed-balance conditions:
\begin{equation}\label{eqn:random_bath_second_order_thermal_condition}
\frac{\int_{-\infty}^{\infty} u(\omega;\lambda) d\mu(\omega)}{\int_{-\infty}^{\infty} d(\omega;\lambda) d\mu(\omega)} = \exp(-\lambda)\,.
\end{equation}
Different from zero-bath case, thermal state preparation in this case does not require $f(t)$ to satisfy the thermal equilibrium condition for each individual realization of $\omega$; instead, it suffices for the condition to hold after averaging over the random frequency $\omega$. Unfortunately, in the following theorem, we will demonstrate that this flexibility does not necessarily lead to an improved complexity:
\begin{thm}\label{thm:random_bath_lower_bound}
For $N\ge2$, set the system bath interaction algorithm with~\cref{assump:zero_bath} (random-bath) and $f(t)\in L^1(\mathbb{R})$ such that
$\mathbb{E}_{\omega}\left(u(\omega;\lambda)\right)>0$ and $\mathbb{E}_{\omega}\left(d(\omega;\lambda)\right)>0$ for all $\lambda\in[-1,1]$. If, for some $\lambda\in[-1,1]$, $u(\omega;\lambda)\ne e^{-\lambda}d(\omega;\lambda)$ on a set of positive $\mu$-measure, then, whenever $T^{-1}$, $\alpha$, and $\epsilon$ are smaller than sufficiently small constants independent of $N$ and
$\int_{|t|>T/2}|f(t)|\mathrm{d}t=\mathcal{O}(\alpha^2)$, achieving $\epsilon$-accuracy for thermal state preparation requires, in the worst case over this class of diagonal Hamiltonians,
$T_{\rm total}=\Omega(N/\epsilon)$.
\end{thm}

The rigorous statement and proof of the above theorem are provided in~\cref{thm:random_bath_lower_bound_rigor} of~\cref{sec:state_random_bath}. This result implies that, to overcome the linear-in-$N$ complexity barrier, the filter function $\hat{f}$ and $p_{\rm b}$ must be designed so that the detailed balance conditions hold $\mu$-almost surely. Consequently, despite the additional flexibility introduced by the random bath frequencies, the design constraints on $f$ and $p_{\rm b}$ for thermal state preparation are essentially similar to those in the zero-bath setting, where $\mu=\delta_0$.

The fundamental bottleneck identified in~\cref{thm:random_bath_lower_bound} arises from the additional correlations among system qubits induced by the randomness of $\omega$. For example, consider $H=\lambda\sum_{i=1}^N c_i^\dagger c_i$. Even when~\eqref{eqn:random_bath_second_order_thermal_condition} is satisfied, the algorithm guarantees only that each single-qubit subsystem has the correct thermal marginal; it does not ensure that the joint state is a product state, as required for the target Gibbs state. In contrast, this obstruction is absent in ground-state preparation, where the target is a pure product state and correct single-qubit marginals uniquely determine the global state. As a result, ground-state preparation can fully exploit the additional flexibility provided by the random bath frequencies, achieving polylogarithmic complexity in both $N$ and $1/\epsilon$, as established in~\cref{thm:main_ground}.

{\it Conclusion and outlook} --
We have shown that an extensive engineered bath enables dissipative preparation of free-fermion thermal and ground states with sublinear total Hamiltonian simulation time, while retaining local system-bath couplings and convergence from arbitrary initial states. Specifically, the thermal state preparation algorithm achieves $\widetilde{\mathcal{O}}(\sqrt{N}/\epsilon)$ scaling at fixed inverse temperature, whereas ground state preparation algorithm requires $\widetilde{\mathcal{O}}\left(\Delta^{-2}\mathrm{polylog}(N/\epsilon)\right)$ with $\Delta$ being the spectral gap of the system. Together, these results demonstrate that simultaneous coupling to an extensive bath can yield a provable reduction in simulation time despite the need to control larger errors beyond leading order.

Our lower bound results also expose distinct obstacles to faster thermal state preparation within the protocols studied. In the zero-bath setting, preserving the target state at second order is insufficient: fourth-order corrections must also satisfy a thermal state preservation condition to improve the $\sqrt{N}/\epsilon$ scaling. For a shared random bath frequency, detailed balance only after averaging is likewise insufficient to break the linear-in-$N$ barrier; the second-order transition rates must satisfy detailed balance for almost every sampled frequency. Constructing filters that satisfy these conditions to achieve improved scaling, or proving that no such filters exist, remain open problems.

Very recently,~\cite{yu2026bathassistedcoolingresets} introduced a distinct system-bath approach that achieves a physical evolution time of $\mathrm{polylog}(N/\epsilon)$ for ground state preparation of weakly interacting fermionic systems. Rather than relying on dissipative mixing, the protocol coherently transfers the information and entropy of the initial system state into decoupled bath modes and then adiabatically transports the remaining system-bath state along a ground state path. Thus, although the construction likewise exploits local system-bath interactions, its underlying mechanism is fundamentally different from the relaxation-based mechanism considered here. An interesting direction for future work is to investigate whether some of these ideas—such as the construction of bright and dark bath sectors or the coherent storage of unwanted information in decoupled bath modes—can be incorporated into genuinely dissipative state-preparation schemes, potentially replacing the subsequent adiabatic transport with dissipative relaxation to handle more general interacting systems.

Further directions include quantifying the trade-off between bath size and simulation time for intermediate bath sizes, and determining whether suitable bath correlations or alternative coupling structures can circumvent the identified obstructions. Extending the analysis to weakly interacting fermions would test whether the advantage survives. Free fermions thus provide a controlled setting for identifying which features of engineered dissipation can support faster state preparation in interacting systems.

\section*{Declaration of AI Assistance}
Generative AI tools (ChatGPT 6.0, OpenAI) were used in the development and preparation of this manuscript. The upper-bound results and their proofs were developed independently by the authors, with AI assistance limited to language polishing, checking, and presentation. For the lower-bound results, AI tools contributed ideas and proof strategies, including the central idea underlying the proof of Theorem \ref{thm:zero_bath_lower_bound_new_1}. All AI-assisted arguments were subsequently reformulated, rigorously developed, and independently verified by the authors. The authors take full responsibility for the correctness and originality of all results, proofs, and conclusions presented in this work.

\section*{Acknowledgments}
This material is based upon work supported by the U.S. Department of Energy, Office of Science, Accelerated Research in Quantum Computing Centers, Quantum Utility through Advanced Computational Quantum Algorithms, grant no. DE-SC0025572 (L.L.). L.L. is a Simons Investigator in Mathematics.  This work was also supported in part by the University of Michigan through a startup grant of Z.D. (Z.D., K.W.) and the Van Loo Postdoctoral Fellowship (K.W.).

\bibliographystyle{apsrev4-2}
\bibliography{ref.bib}

\appendix

\setcounter{secnumdepth}{2}
\onecolumngrid
\newpage
\begin{center}
{\large\bfseries Supplementary Material for\\[0.5em]
Sublinear thermal state and ground state preparation for free fermions via system-bath interaction\par}
\end{center}
\vspace{1em}

\section{Zero-bath setting for thermal state preparation}\label{sec:zero_bath_thermal}

In this section, we consider the thermal state preparation problem in the zero-bath setting. We will provide a rigorous statement of Theorem~\ref{thm:main_thermal} for complexity upper bound and Theorem~\ref{thm:zero_bath_lower_bound_rigor} for the complexity lower bound.

We start with the upper bound result. The following theorem provides a rigorous statement of Theorem~\ref{thm:main_thermal}.
\begin{thm}[Zero bath Hamiltonian for thermal state preparation]\label{thm:thermal_zero_bath}
Under~\cref{assump:zero_bath} and $\Lambda:=\max_k|\lambda_k|=\mathcal{O}(1)$, we choose
\begin{equation}\label{eqn:f_thermal_state}
\hat{f}(\nu)=\exp\left(-\sqrt{1+\beta^2\nu^2}/4\right)\omega(\nu/S)e^{-\beta\nu/4}\,,
\end{equation}
where $S=\Lambda+1$ and $\omega$ is a Gevrey function of order $s>1$ with compact support that is nonnegative, positive for $|\nu|<1$, and nonincreasing for $\nu\ge0$, such that $\omega(\nu)=\omega(-\nu)$ and $\omega(\nu)=1$ for $\nu\in[-1/2,1/2]$ and $\omega(\nu)=0$ for $|\nu|>1$. When $T=\widetilde{\Theta}(1+\beta)$ and $\alpha=\widetilde{\Theta}(\epsilon^{1/2}/(\log^{2}(2+\beta S)\min\{N^{1/4}\exp(\beta\Lambda/4),N^{1/2}\}))$, we have
\[
\left\|\rho_{\rm fix}(\Phi_{\alpha})-\sigma_{\beta}\right\|_1<\epsilon\,,
\]
and
\[
\sup_{\rho}\left\|\Phi_{\alpha}^m(\rho)-\sigma_{\beta}\right\|_1\leq \epsilon,\quad \text{when}\quad m=\widetilde{\Theta}\left(\frac{\min\{N^{1/2}\exp(\beta\Lambda/2),N\}}{\epsilon}\right)\,.
\]
In particular, when $\beta=\Theta(1)$, to prepare the Gibbs state $\sigma_\beta$ with $\epsilon$-accuracy, the total Hamiltonian simulation time is $T_{\rm total}=\widetilde{\mathcal{O}}(\beta N^{1/2}/\epsilon)$. Here, $\widetilde{\Theta},\widetilde{\mathcal{O}}$ hide some polylogarithmic factors of $\beta,N,1/\epsilon$.
\end{thm}
We put the proof of the above theorem in~\cref{sec:proof_thermal}. In~\cref{thm:thermal_zero_bath}, $f$ is chosen to satisfy the condition $\hat{f}(\nu)\exp(\beta\nu/2)=\hat{f}(-\nu)$. The existence of such a filter function is guaranteed by the construction in~\cite{Ding_2025}. According to~\cite[Section 3]{Ding_2025}, this guarantees that the approximate Lindbladian operator for $\Phi_{\alpha}$ can satisfy the KMS detailed balance condition with respect to the Gibbs state $\sigma_\beta$ up to a coherent term that commutes with the system Hamiltonian $H$.

\subsection{Lower bound for zero-bath setting}\label{sec:zero_bath_lower_bound}

In this section, we discuss the lower bound result for thermal state preparation in the zero-bath setting. As mentioned in main text, we assume $\beta=1$ and $\|h\|\leq 1$ for simplicity. We first provide conditions on the filter function $f(t)$ for achieving improved complexity.

\begin{condition}[Improvement condition for zero-bath setting]\label{cond:zero_bath_improvement}
Define $\hat{f}(\nu)=\int^\infty_{-\infty}f(t)\exp(i\nu t)\mathrm{d}t$. Given $\lambda\in [-1,1]$, we say $f$ satisfies \emph{the $\lambda$-improvement condition} if the following two conditions hold:
\begin{itemize}
\item (Second order thermal state preservation)
\begin{equation}\label{eqn:zero_bath_second_order_thermal_condition}
\frac{\left|\hat{f}(\lambda)\right|^2}{\left|\hat{f}(\lambda)\right|^2+\left|\hat{f}(-\lambda)\right|^2}=p(\lambda):=\frac{1}{1+\exp(\lambda)}
\end{equation}
\item (Fourth order thermal state preservation)
\begin{equation}\label{eqn:zero_bath_fourth_order_thermal_condition}
\bigl(1-p(\lambda)\bigr)G(\lambda)
-p(\lambda)G(-\lambda)=0,
\end{equation}
where $G(\lambda)=-2\operatorname{Re}
   \left[\overline{\hat{f}(\lambda)}J(\lambda)\right]$ with
\[
\begin{aligned}
J(\nu):=\int_{-\infty}^{\infty}
   \int_{-\infty}^{t_1}
   \int_{-\infty}^{t_2}\,
   e^{i\nu(t_1-t_2+t_3)}f(t_3)\left[f(t_1)\overline{f(t_2)}
   +\overline{f(t_1)}f(t_2)\right]dt_3dt_2dt_1\,.
\end{aligned}
\]
In particular, in the case~\eqref{eqn:zero_bath_second_order_thermal_condition} is satisfied,~\eqref{eqn:zero_bath_fourth_order_thermal_condition} is equivalent to $\operatorname{Re}\frac{J(\lambda)}{\hat f(\lambda)}
=\operatorname{Re}\frac{J(-\lambda)}{\hat f(-\lambda)}$.
\end{itemize}
\end{condition}

The following theorem provides a necessary and sufficient condition for achieving improved complexity in the zero-bath setting for general free fermion Hamiltonians:
\begin{thm}[Criterion for improved complexity in the zero-bath setting]\label{thm:zero_bath_lower_bound_rigor}
Consider the system-bath interaction algorithm under~\cref{assump:zero_bath} at inverse temperature $\beta=1$. Let $f\in L^1(\mathbb R)$ be a fixed filter satisfying
\[
|\hat f(\lambda)|^2+|\hat f(-\lambda)|^2>0
\quad\text{for all }\lambda\in[-1,1],
\qquad
\eta(T):=\int_{|t|>T/2}|f(t)|\,\mathrm dt
=\mathcal O(e^{-T^s})
\]
for some $s>0$.

Fix $C_{\rm tail}>0$. There exist constants $\alpha_0,T_0>0$ and $\epsilon_0\in(0,1)$, depending only on $f$ and $C_{\rm tail}$ such that if the parameters $\alpha,T,\epsilon$ satisfy
\begin{equation}\label{eqn:condition_regime}
0<\alpha\le\alpha_0,\qquad
T\ge T_0,\qquad
0<\epsilon\le\epsilon_0,\qquad
\eta(T)\le C_{\rm tail}\alpha^4,
\end{equation}
then the following two conditions are equivalent:
\begin{enumerate}
\item The filter $f$ satisfies the $\lambda$-improvement condition in~\cref{cond:zero_bath_improvement} for every $\lambda\in[-1,1]$.
\item For every $N$-mode free-fermion Hamiltonian $H=\sum_{i,j}h_{ij}c_i^\dagger c_j$ with $\|h\|\le1$ and every $0<\epsilon\le\epsilon_0$, the parameters $\alpha,T$ can be chosen in the above regime, together with a number of iterations $m$, so that
\[
\sup_\rho\|\Phi_\alpha^m(\rho)-\sigma_1(H)\|_1\le\epsilon,
\qquad
T_{\rm total}=2Tm
=\widetilde{\mathcal O}\!\left(\frac{N^{1/4}}{\epsilon^{1/2}}\right).
\]
\end{enumerate}
Moreover, if condition~1 fails, there exist an energy $\lambda\in[-1,1]$ and a constant $c>0$ such that, for $H=\lambda\sum_{i=1}^N c_i^\dagger c_i$, any choice of parameters that satisfy ~\eqref{eqn:condition_regime} and achieve $\epsilon$-accuracy uniformly over all initial states must satisfy
\[
T_{\rm total}\ge c\,\frac{\sqrt N}{\epsilon}.
\]
\end{thm}

We put the proof of~\cref{thm:zero_bath_lower_bound_rigor} in Appendix~\ref{sec:proof_thermal_lower_bound}. In Condition~\ref{cond:zero_bath_improvement}, the second-order thermal state preservation condition~\eqref{eqn:zero_bath_second_order_thermal_condition} is necessary to ensure that the leading $\mathcal{O}(\alpha^2)$ contribution to the channel preserves the target thermal state. Although this contribution vanishes as $\alpha\to 0$, failure to satisfy the condition leaves an $\mathcal{O}(1)$ discrepancy between the fixed point and the target thermal state in this limit. The algorithm therefore cannot achieve arbitrarily small error $\epsilon$ in the weak-coupling regime considered here. When the second-order condition holds, the fourth-order thermal state preservation condition~\eqref{eqn:zero_bath_fourth_order_thermal_condition} is necessary and sufficient for the fixed-point error to be $\mathcal{O}(\alpha^4)$. In this case, $\epsilon$-accuracy can be achieved by choosing $\alpha=\Theta(\epsilon^{1/4}N^{-1/8})$ and $m=\widetilde{\mathcal O}(N^{1/4}/\epsilon^{1/2})$, yielding a total Hamiltonian simulation time of $T_{\rm total}=\widetilde{\mathcal O}(N^{1/4}/\epsilon^{1/2})$ according to the analysis in~\cref{sec:proof_thermal}. In addition, for general free-fermion Hamiltonians, we can also transform the system Hamiltonian to a diagonal form and apply the same analysis to $H=\sum_{i=1}^N\lambda_in_i$ with $\lambda_i\in[-1,1]$, meaning this improved scaling is achievable for general free-fermion Hamiltonians. Thus, together with the second-order condition, the fourth-order condition is sufficient to improve upon the $\sqrt{N}/\epsilon$ scaling for thermal state preparation. However, it remains unclear whether a filter function exists that satisfies both~\eqref{eqn:zero_bath_second_order_thermal_condition} and~\eqref{eqn:zero_bath_fourth_order_thermal_condition} for all $\lambda\in[-1,1]$ and we leave this as an open question for future work.

Together with~\cref{thm:zero_bath_lower_bound_rigor}, the following lemma yields~\cref{thm:zero_bath_lower_bound_new_1}:
\begin{lem}\label{lem:zero_bath_lower_bound_rigor}
Let $1<S\le3$. Let $w\in C_c^\infty([-1,1])$ be even, nonnegative, positive on $(-1,1)$, and nonincreasing for $x\ge0$. Then the filter with $\hat f(\nu)=F_S(\nu):=e^{-\nu/4}e^{-\sqrt{1+\nu^2}/4}w(\nu/S)$ cannot satisfy Condition~\ref{cond:zero_bath_improvement}~\eqref{eqn:zero_bath_fourth_order_thermal_condition} on $[-1,1]$.
\end{lem}
We put the proof of~\cref{lem:zero_bath_lower_bound_rigor} in Appendix~\ref{sec:proof_thermal_lower_bound}. For the monotone cutoff chosen in~\cref{thm:thermal_zero_bath}, this lemma and~\cref{thm:zero_bath_lower_bound_rigor} imply~\cref{thm:zero_bath_lower_bound_new_1}, as explained at the end of the proof.

\subsection{Proof of Theorem~\ref{thm:thermal_zero_bath}}\label{sec:proof_thermal}

To prove~\cref{thm:thermal_zero_bath}, we first have the following proposition:
\begin{prop}\label{prop:zero_bath}
Under~\cref{assump:zero_bath}, there are CPTP maps
$\Phi_{\alpha,j}:\mathbb C^{2^N\times2^N}\to
\mathbb C^{2^N\times2^N}$ which act on the $j$th system mode such that
\begin{itemize}
\item
$
[\Phi_{\alpha,j},\Phi_{\alpha,k}]=0
\quad(j\ne k),
\qquad \Phi_\alpha
=\Phi_{\alpha,1}\circ\cdots\circ\Phi_{\alpha,N}.
$
\item $\Phi_{\alpha,j}$ maps diagonal terms to diagonal terms and off-diagonal terms to off-diagonal terms in the basis of $d_j$. Define $n_j=d^\dagger_jd_j$. We have
\begin{equation}\label{eqn:exact_fermionic_mode_channel}
\Phi_{\alpha,j}(\rho)
=\sum_{r\in\{00,10,01,11\}}K_{j,r}\rho K_{j,r}^{\dagger},
\qquad
\sum_{r\in\{00,10,01,11\}}K_{j,r}^{\dagger}K_{j,r}=I,
\end{equation}
where the four Kraus operators take the form:
\begin{equation}\label{eqn:exact_fermionic_mode_kraus}
\begin{aligned}
K_{j,00}&=a_{j,00}^{(0)}(I-n_j)+a_{j,00}^{(1)}n_j,
&K_{j,10}&=a_{j,10}d_j^\dagger,\\
K_{j,01}&=a_{j,01}d_j,
&K_{j,11}&=a_{j,11}^{(0)}(I-n_j)+a_{j,11}^{(1)}n_j.
\end{aligned}
\end{equation}
\end{itemize}
In addition, when $\alpha\|f\|_{L^1}=\mathcal{O}(1)$, there exists $\mc{L}_j$ such that
\begin{equation}\label{eqn:Phi_alpha_j}
\left\|\Phi_{\alpha,j}-\mathcal U_j(T)\circ\exp\left(\mathcal L_{j}\alpha^2\right)\circ \mathcal U_j(T)\right\|_{1\leftrightarrow1}=\mc{O}\left(\alpha^4\min\left\{\left\|f(t)\right\|^4_{L^\infty}T^4,\|f\|_{L^1}^4\right\}\right)\,,
\end{equation}
where $\mathcal U_j(T)[\rho]
=e^{-i\lambda_jn_jT}\rho e^{i\lambda_jn_jT}$. Here, the Lindbladian
\begin{equation}\label{eqn:embedded_fermionic_lindbladian}
\begin{aligned}
\mathcal L_j(\rho)
={}&-i[G_{j,\rm Lamb},\rho]+L_{j,\rm c}\rho L_{j,\rm c}^\dagger
-\frac12\left\{L_{j,\rm c}^\dagger L_{j,\rm c},\rho\right\}+L_{j,\rm n}\rho L_{j,\rm n}^\dagger
-\frac12\left\{L_{j,\rm n}^\dagger L_{j,\rm n},\rho\right\},
\end{aligned}
\end{equation}
with jump operators
\begin{equation}\label{eqn:embedded_fermionic_jump_operators}
L_{j,\rm c}=\int^{T}_{-T}f(t)e^{i\lambda_jt}\mathrm{d}td_j^\dagger,
\qquad
L_{j,\rm n}=\int^{T}_{-T}f(t)e^{-i\lambda_jt}\mathrm{d}td_j,
\end{equation}
and $[G_{j,\rm Lamb},\sigma_\beta]=0$.
\end{prop}
\begin{proof}[Proof of~\cref{prop:zero_bath}]
First, we notice that
\[
c_i^\dagger=\sum_j U^*_{i,j}d_j^\dagger,\qquad c_i=\sum_j U_{i,j}d_j\,.
\]
Define two new sets of bath operators:
\[
e^\dagger_{j,+} = \sum_{i=1}^N U^*_{i,j}b^\dagger_{i,+},\qquad e^\dagger_{j,-} = \sum_{i=1}^N U_{i,j}b^\dagger_{i,-}\,.
\]
Because $U$ is unitary, the operators $\{e^\dagger_{j,+},e_{j,+}\}^N_{j=1}$ and $\{e^\dagger_{j,-},e_{j,-}\}^N_{j=1}$ jointly satisfy the canonical anti-commutation relations. In addition, the bath vacuum remains the same product vacuum state in the new bath modes, and the new bath operators still commute with system operators such as $\{d^\dagger_j,d_j\}$. From~\cref{assump:zero_bath}, we then have
\[
H_{\rm int}(t)=\sum_{j=1}^N\left[f(t)\left(d_j^\dagger\otimes e_{j,+}^\dagger+d_j\otimes e_{j,-}^\dagger\right)+\overline{f(t)}\left(d_j\otimes e_{j,+}+d_j^\dagger\otimes e_{j,-}\right)\right]\,.
\]

We then have
\begin{equation}\label{eqn:fermionic_energy_mode_hamiltonians}
H_{\alpha}(t)=\sum_{j=1}^NH_{\alpha,j}(t),
\end{equation}
where
\begin{equation}\label{eqn:fermionic_energy_mode_hamiltonian}
H_{\alpha,j}(t)
=\lambda_jn_j+\alpha\left[f(t)\left(
d_j^\dagger\otimes e_{j,+}^\dagger
+d_j\otimes e_{j,-}^\dagger\right)
+\overline{f(t)}\left(
d_j\otimes e_{j,+}
+d_j^\dagger\otimes e_{j,-}\right)\right].
\end{equation}
Every interaction term in $H_{\alpha,j}$ is odd in the system and
odd in the bath.  This implies that
\begin{equation}\label{eqn:fermionic_energy_mode_commutation}
[H_{\alpha,j}(s),H_{\alpha,k}(t)]=0
\qquad(j\ne k).
\end{equation}
Consequently,
\begin{equation}\label{eqn:fermionic_joint_unitary_product}
U^\alpha(t)=\prod_{j=1}^NU_{\alpha,j}(t),
\quad
U_{\alpha,j}(t):=\mathcal T\exp\left(
-i\int_{-T}^tH_{\alpha,j}(s)\,\mathrm ds\right),\quad [U_{\alpha,j}(t),U_{\alpha,k}(s)]=0\quad(j\ne k)\,.
\end{equation}

Now, we claim that the reduced channel $\Phi_\alpha$ in~\eqref{eqn:Phi_alpha} can be decomposed as
\begin{equation}\label{eqn:fermionic_channel_decomposition}
\Phi_\alpha=\Phi_{\alpha,1}\circ\cdots\circ\Phi_{\alpha,N},
\end{equation}
where $\Phi_{\alpha,j}$ denotes the reduced channel induced by $U_{\alpha,j}(t)$ after tracing out the bath, i.e.,
\[
\Phi_{\alpha,j}(\rho)=\operatorname{Tr}_E\left[U_{\alpha,j}(T)(\rho\otimes \rho_E)U_{\alpha,j}(T)^\dagger\right]\,.
\]
To derive~\eqref{eqn:fermionic_channel_decomposition}, we first note that $U_{\alpha,j}(t)$ acts on different system and bath modes for different $j$. We then expand each time-evolution operators $U_{\alpha,j}(t)$ into Dyson series. Because trace out only keeps the terms with
$\ket{0}\bra{0}$ or $\ket{1}\bra{1}$ on every bath mode and the bath begins in the vacuum, we can replace all creation and annihilation operators on the bath modes by $\ket{1}\bra{0}$ and $\ket{0}\bra{1}$, respectively. Therefore, we can separately trace out the bath modes for each $j$ and obtain~\eqref{eqn:fermionic_channel_decomposition}.

Finally, for each $\Phi_{\alpha,j}$, applying the weak-coupling expansion in~\cite[Theorem 11]{ding2025endtoendefficientquantumthermal} and~\cite[Theorem S1]{wang2025lindblad} gives~\eqref{eqn:Phi_alpha_j}. The interaction terms $c^\dagger_i\otimes b^\dagger_{i,+}$ and $c_i\otimes b^\dagger_{i,-}$ generate $L_{j,\rm c}$ and $L_{j,\rm n}$ in~\cref{prop:zero_bath}, respectively. Moreover, the Lamb-shift terms are linear combination of $1-n_j$ and $n_j$, so $[G_{j,\rm Lamb},n_k]=0$ for all $k$. This concludes the proof.
\end{proof}

Now, we are ready to prove~\cref{thm:thermal_zero_bath}.

\begin{proof}[Proof of~\cref{thm:thermal_zero_bath}] We first note that the choice of $f$ in~\cref{thm:thermal_zero_bath} satisfies~\cite[Assumption 13]{Ding_2025} (see~\cite[(3.20)]{Ding_2025}). According to~\cite[Section 3.1]{Ding_2025}, for any $j$, there exists a coherent term $B_j$ such that
\[
\mathcal{L}_{j,\rm KMS}(\rho):=-i[B_j,\rho]+\left( L_{j,\rm c,\infty} \rho L_{j,\rm c,\infty}^\dagger - \frac{1}{2} \left\{ L_{j,\rm c,\infty}^\dagger L_{j,\rm c,\infty}, \rho \right\} \right)+\left( L_{j,\rm n,\infty} \rho L_{j,\rm n,\infty}^\dagger - \frac{1}{2} \left\{ L_{j,\rm n,\infty}^\dagger L_{j,\rm n,\infty}, \rho \right\} \right)
\]
satisfies KMS-DBC with respect to the Gibbs state $\sigma_\beta=\prod_j \sigma_{\beta,j}$, where $\sigma_{\beta,j}=\frac{1}{1+\exp(-\beta\lambda_j)}(1-n_j)+\frac{\exp(-\beta\lambda_j)}{1+\exp(-\beta\lambda_j)}n_j$, i.e., $\mathcal{L}_{j,\rm KMS}(\sigma_\beta)=0$. Here,
\[
L_{j,\rm c,\infty}=\int^{\infty}_{-\infty}f(t)\exp(i\lambda_jt)d^\dagger_j\mathrm{d}t,\quad L_{j,\rm n,\infty}=\int^{\infty}_{-\infty}f(t)\exp(-i\lambda_jt)d_j\mathrm{d}t\,.
\]
In addition, according to the formula of $B_{j}$ in~\cite[Theorem 10]{Ding_2025}, we have $[B_{j},n_k]=0$ for any $k$.

Furthermore, according to~\cite[Lemma 30, (D3, D13)]{Ding_2025}, we have $\|f(t)\|_{L^\infty}=\mathcal{O}(1)$, $\|f(t)\|_{L^1}=\mathcal{O}\left(\log(2+\beta S)\right)$ and
\[
\int_{|t|>T}|f(t)|\mathrm{d}t=\mathcal{O}\left((1+\beta)\exp\left(-\Theta\left((T/(1+\beta))^{1/s}\right)\right)\right)\,.
\]
Because $\|d^\dagger_j\|=\|d_j\|=1$, this implies that
\[
\left\|\mathcal{L}_{j,\rm KMS}+i[B_j-G_{j, \rm Lamb},\cdot]-\mathcal{L}_{j}\right\|_{1\leftrightarrow1}=\mathcal{O}\left((1+\beta)\exp\left(-\Theta\left((T/(1+\beta))^{1/s}\right)\right)\right)\,.
\]
Furthermore, according to~\cref{prop:zero_bath}~\eqref{eqn:Phi_alpha_j}, we have
\begin{equation}\label{eqn:Phi_alpha_j_approx}
\begin{aligned}
&\left\|\Phi_{\alpha,j}-\underbrace{\mathcal{U}_j(T)\circ\exp\left(\left(\mathcal{L}_{j,\rm KMS}+i[B_j-G_{j, \rm Lamb},\cdot]\right)\alpha^2\right)\circ \mathcal{U}_j(T)}_{:=\widetilde{\Phi}_{\alpha,j}}\right\|_{1\leftrightarrow1}\\
=&\mc{O}\left(\alpha^2(1+\beta)\exp\left(-\Theta\left((T/(1+\beta))^{1/s}\right)\right)+\alpha^4\|f(t)\|_{L^1}^4\right)\\
=&\mc{O}\left(\alpha^2(1+\beta)\exp\left(-\Theta\left((T/(1+\beta))^{1/s}\right)\right)+\alpha^4\log^4(2+\beta S)\right)
\end{aligned}
\end{equation}
It is straightforward to check that $\widetilde{\Phi}_{\alpha,j}(\sigma_{\beta})=\sigma_{\beta}$ for each $j$. In addition, $\tilde{\Phi}_{\alpha,j}$ also takes the form of~\cref{prop:zero_bath}~\eqref{eqn:exact_fermionic_mode_channel} such that
\begin{equation}\label{eqn:exact_fermionic_mode_channel_perfect}
\widetilde{\Phi}_{\alpha,j}(\rho)
=\sum_{r\in\{00,10,01,11\}}\widetilde{K}_{j,r}\rho \widetilde{K}_{j,r}^{\dagger},
\qquad
\sum_{r\in\{00,10,01,11\}}\widetilde{K}_{j,r}^{\dagger}\widetilde{K}_{j,r}=I,
\end{equation}
where the four Kraus operators have the form
\begin{equation}\label{eqn:exact_fermionic_mode_kraus_perfect}
\begin{aligned}
\widetilde{K}_{j,00}&=\widetilde{a}_{j,00}^{(0)}(I-n_j)+\widetilde{a}_{j,00}^{(1)}n_j,
&\widetilde{K}_{j,10}&=\widetilde{a}_{j,10}d_j^\dagger,\\
\widetilde{K}_{j,01}&=\widetilde{a}_{j,01}d_j,
&\widetilde{K}_{j,11}&=\widetilde{a}_{j,11}^{(0)}(I-n_j)+\widetilde{a}_{j,11}^{(1)}n_j.
\end{aligned}
\end{equation}
According to~\cref{prop:zero_bath} and~\eqref{eqn:Phi_alpha_j_approx} and the eigenvector perturbation theorem, we can set $T=\Omega((1+\beta)\mathrm{polylog}(2+\beta))$ and $\alpha=\widetilde{\mathcal{O}}(\log^{-2}(2+\beta S))$ so that
\[
\rho_{\rm fix}\left(\Phi_\alpha\right)=\prod^N_{j=1}\widetilde{\sigma}_{\beta,j},\ \text{such that}\quad \widetilde{\sigma}_{\beta,j}=\widetilde{p}_jn_j+(1-\widetilde{p}_j)(I-n_j),\quad \Phi_{\alpha,j}(\widetilde{\sigma}_{\beta,j})=\widetilde{\sigma}_{\beta,j},\quad 0\leq \widetilde{p}_j\leq 1
\]
This means that the fixed point of $\Phi_\alpha$ is a product state, and the $j$-th mode of $\rho_{\rm fix}\left(\Phi_\alpha\right)$ does not have off-diagonal terms in the $n_j$ basis.

Furthermore, for any $1\leq j\leq N$, we define two quantum channels
\[
\mathcal{R}_j(\rho)=\left(\prod^j_{i=1}\sigma_{\beta,i}\right)\left(I_{1:j}\otimes \mathrm{Tr}_{1:j}(\rho)\right),\quad \widetilde{\mathcal{R}}_j(\rho)=\left(\prod^j_{i=1}\widetilde{\sigma}_{\beta,i}\right)\left(I_{1:j}\otimes \mathrm{Tr}_{1:j}(\rho)\right)\,,
\]
where $\mathrm{Tr}_{1:j}$ denotes the partial trace over the first $j$ modes and $I_{1:j}$ denotes the identity operator on the first $j$ modes. We also define $\mc{R}_0=\mathrm{id}$ and $\widetilde{\mc{R}}_0=\mathrm{id}$.

Given any $\rho$,
\[
\widetilde{\Phi}^m_{\alpha,j}\left(\mathcal{R}_{j-1}(\rho)\right)=\left(\prod^{j-1}_{i=1}\sigma_{\beta,i}\right)\widetilde{\Phi}^m_{\alpha,j}\left(I_{1:j-1}\otimes \mathrm{Tr}_{1:j-1}(\rho)\right)\,.
\]
We note that $\widetilde{\Phi}^m_{\alpha,j}$ acts on $I_{1:j-1}\otimes \mathrm{Tr}_{1:j-1}(\rho)$ is same as a local quantum channel that only nontrivially acts on the $j$-th qubit with $d^\dagger_j$ replaced by $\ket{1}\bra{0}$ and $d_j$ replaced by $\ket{0}\bra{1}$. Therefore,there exists a uniform constant $\kappa=\Omega(1)$ with $\kappa<1$ such that
\[
\sup_j\left\|\widetilde{\Phi}^m_{\alpha,j}\left(\mathcal{R}_{j-1}(\cdot)\right)-\mathcal{R}_{j}\right\|_{1\leftrightarrow1}=\mathcal{O}\left(\frac{\epsilon}{N}\right),\quad\text{when}\quad m=\Omega\left(\alpha^{-2}\kappa^{-1}\mathrm{polylog}(N/\epsilon)\right)\,.
\]
According to the eigenvalue perturbation theorem and~\cref{eqn:Phi_alpha_j_approx}, when $T=\Omega((1+\beta)\mathrm{polylog}(2+\beta/\kappa))$ and $\alpha=\widetilde{\mathcal{O}}(\kappa^{1/2}/\log^{2}(2+\beta S))$, we also have
\[
\sup_j\left\|\Phi^{m}_{\alpha,j}\left(\widetilde{\mathcal{R}}_{j-1}(\cdot)\right)-\widetilde{\mathcal{R}}_{j}\right\|_{1\leftrightarrow1}=\mathcal{O}\left(\frac{\epsilon}{N}\right),\quad\text{when}\quad m=\Omega\left(\alpha^{-2}\kappa^{-1}\mathrm{polylog}(N/\epsilon)\right)\,.
\]
This implies that
\begin{equation}\label{eqn:zero_bath_mixing_times_bound}
\begin{aligned}
&\left\|\Phi_{\alpha,N}^m\circ \dots \circ \Phi_{\alpha,1}^m(\rho)-\rho_{\rm fix}\left(\Phi_\alpha\right)\right\|_1\\
\leq &\sum^N_{j=1}\left\|\Phi_{\alpha,N}^m\circ \dots \Phi_{\alpha,j}^m\left(\widetilde{\mc R}_{j-1}(\rho)\right)-\Phi_{\alpha,N}^m\circ \dots \Phi_{\alpha,j+1}^m\left(\widetilde{\mc R}_{j}(\rho)\right)\right\|_1\\
\leq &\sum^N_{j=1}\left\|\Phi^{m}_{\alpha,j}\left(\widetilde{\mathcal{R}}_{j-1}(\rho)\right)-\widetilde{\mathcal{R}}_{j}(\rho)\right\|_{1}=\mathcal{O}\left(\epsilon\right),\ \text{when}\ m=\Omega\left(\alpha^{-2}\kappa^{-1}\mathrm{polylog}(N/\epsilon)\right)\,.
\end{aligned}
\end{equation}
for any density operator $\rho$.

Finally, according to~\cite[Theorem 8]{ding2025endtoendefficientquantumthermal} and~\cref{eqn:Phi_alpha_j_approx}, setting $T=\Theta((1+\beta)\mathrm{polylog}(N(1+\beta)/\epsilon))$ and $\alpha=\widetilde{\Theta}(\epsilon^{1/2}/(\log^{2}(2+\beta S)\min\{N^{1/4}\exp(\beta\Lambda/4),N^{1/2}\}))$ gives
\[
\sup_j\left\|
\begin{pmatrix}
\frac{1}{1+\exp(-\beta\lambda_j)}&0\\
0&\frac{\exp(-\beta\lambda_j)}{1+\exp(-\beta\lambda_j)}
\end{pmatrix}
-
\begin{pmatrix}
1-\widetilde p_j&0\\
0&\widetilde p_j
\end{pmatrix}\right\|_1=\mathcal{O}\left(\epsilon\max\left\{\frac{1}{\sqrt{\exp(\beta \Lambda)N}},\frac{1}{N}\right\}\right)\,.
\]
Since $\Lambda:=\max_k|\lambda_k|=\mathcal{O}(1)$, we have
$\begin{pmatrix}
\frac{1}{1+\exp(-\beta\lambda_j)}&0\\
0&\frac{\exp(-\beta\lambda_j)}{1+\exp(-\beta\lambda_j)}
\end{pmatrix}\geq CI$, where $C=\Omega(\exp(-\beta\Lambda))$ is a constant independent of $N$ and $j$. According to~\cref{lem:product_state_error}, after tuning the constants in $T,\alpha$ properly, we have
\[
\left\|\rho_{\rm fix}(\Phi_{\alpha})-\sigma_{\beta}\right\|_1\leq \frac{\epsilon}{2}\,.
\]
Combining this and~\eqref{eqn:zero_bath_mixing_times_bound}, we obtain
\[
\sup_{\rho}\left\|\Phi_{\alpha}^m(\rho)-\sigma_{\beta}\right\|_1\leq \epsilon,\quad \text{when}\quad m=\widetilde{\Theta}\left(\frac{\min\{N^{1/2}\exp(\beta\Lambda/2),N\}}{\epsilon}\right)\,.
\]
This concludes the proof of~\cref{thm:thermal_zero_bath}.
\end{proof}

\subsection{Proof of Theorem~\ref{thm:zero_bath_lower_bound_rigor} and Lemma~\ref{lem:zero_bath_lower_bound_rigor}}\label{sec:proof_thermal_lower_bound}

By~\cref{prop:zero_bath} and~\eqref{eqn:exact_fermionic_mode_kraus},
\[
\Phi_\alpha=\Phi_{\alpha,1}\circ\cdots\circ\Phi_{\alpha,N},
\qquad
\Phi_{\alpha,j}(\rho)
=\sum_{\mu\in\{00,10,01,11\}}K_{j,\mu}\rho K_{j,\mu}^\dagger.
\]
Following the notation in~\cref{prop:zero_bath}, we define
\[
\begin{gathered}
a_j:=a_{\alpha,T}(\lambda_j)=|a_{j,10}|^2,
\qquad
z_j:=z_{\alpha,T}(\lambda_j)=|a_{j,01}|^2,
\qquad
q_j:=q_{\alpha,T}(\lambda_j)=a_j+z_j.
\end{gathered}
\]
Trace preservation gives
\[
\sum_{\mu\in\{00,11\}}|a_{j,\mu}^{(0)}|^2=1-a_j,
\qquad
\sum_{\mu\in\{00,11\}}|a_{j,\mu}^{(1)}|^2=1-z_j.
\]
We first consider the case when the density operator $\rho$ is diagonal in $\ketbra{0}{0},\ketbra{1}{1}$ basis of $d_j$. Define
\[
\begin{gathered}
v(r):=(1-r)\ket{0}\bra{0}+r\ket{1}\bra{1},
\quad 0\le r\le1,
\qquad
\rho_{\mathbf r}:=\bigotimes_{k=1}^Nv(r_k),\\
\rho_{<j}:=\bigotimes_{k<j}v(r_k),
\qquad
\rho_{>j}:=\bigotimes_{k>j}v(r_k).
\end{gathered}
\]
Since
\[
d_j=Z^{\otimes(j-1)}\otimes\ket{0}\bra{1}\otimes I^{\otimes(N-j)},
\qquad
Z:=\ket{0}\bra{0}-\ket{1}\bra{1},
\qquad
Zv(r)Z=v(r),
\]
we obtain, with $n_j=d_j^\dagger d_j$,
\[
\begin{aligned}
\Phi_{\alpha,j}(\rho_{\mathbf r})
&=(1-a_j)(I-n_j)\rho_{\mathbf r}(I-n_j)
 +(1-z_j)n_j\rho_{\mathbf r}n_j+a_jd_j^\dagger\rho_{\mathbf r}d_j
 +z_jd_j\rho_{\mathbf r}d_j^\dagger\\
&=\rho_{<j}\otimes
\begin{pmatrix}
(1-a_j)(1-r_j)+z_jr_j&0\\
0&a_j(1-r_j)+(1-z_j)r_j
\end{pmatrix}
\otimes\rho_{>j}\\
&=\rho_{<j}\otimes v(r_j')\otimes\rho_{>j},
\qquad
r_j'=a_j(1-r_j)+(1-z_j)r_j=a_j+(1-q_j)r_j.
\end{aligned}
\]
Consequently,
\begin{equation}\label{eq:zero_bath_population_update}
\Phi_\alpha\!\left(\bigotimes_{j=1}^N v(r_j)\right)
=\bigotimes_{j=1}^N v(r_j')
=\bigotimes_{j=1}^N
v\!\left(a_{\alpha,T}(\lambda_j)+(1-q_{\alpha,T}(\lambda_j))r_j\right).
\end{equation}
For sufficiently large $T$ and small $\alpha$,~\eqref{eq:zero_bath_transition_expansion} below gives $q_{\alpha,T}(\lambda)=\Theta(\alpha^2)>0$ uniformly on $[-1,1]$. Thus
\[
p_{\alpha,T}(\lambda):=\frac{a_{\alpha,T}(\lambda)}{q_{\alpha,T}(\lambda)}
\quad\Longrightarrow\quad
a_{\alpha,T}(\lambda)+(1-q_{\alpha,T}(\lambda))p_{\alpha,T}(\lambda)
=p_{\alpha,T}(\lambda),
\]
and
\begin{equation}\label{eq:zero_bath_stationary_product}
\rho_*:=\bigotimes_{j=1}^N v(p_{\alpha,T}(\lambda_j)),
\qquad
\Phi_\alpha(\rho_*)=\rho_*.
\end{equation}

Let $F_T(\nu)=\int_{-T}^T f(t)e^{i\nu t}\,\mathrm dt$, let $J_T(\nu)$ be the integral defining $J(\nu)$ in~\cref{cond:zero_bath_improvement} with all three times restricted to $[-T,T]$, and set $G_T(\nu)=-2\operatorname{Re}[\overline{F_T(\nu)}J_T(\nu)]$. A direct Dyson expansion gives
\begin{equation}\label{eq:zero_bath_transition_expansion}
\begin{aligned}
a_{\alpha,T}(\lambda)
&=\alpha^2|F_T(\lambda)|^2+\alpha^4G_T(\lambda)+\mathcal O(\alpha^6),\\
z_{\alpha,T}(\lambda)
&=\alpha^2|F_T(-\lambda)|^2+\alpha^4G_T(-\lambda)+\mathcal O(\alpha^6).
\end{aligned}
\end{equation}
Write $q_2(\lambda)=|\hat f(\lambda)|^2+|\hat f(-\lambda)|^2$, which is bounded below by a positive constant on $[-1,1]$. It follows that
\begin{equation}\label{eq:zero_bath_leading_stationary_occupation}
p_{\alpha,T}(\lambda)
=\frac{|\hat f(\lambda)|^2}{q_2(\lambda)}+\mathcal O(\alpha^2+\eta(T)).
\end{equation}
If~\eqref{eqn:zero_bath_second_order_thermal_condition} holds, define $\Delta_4(\lambda)=(1-p(\lambda))G(\lambda)-p(\lambda)G(-\lambda)$. Expanding the same ratio one order further gives
\begin{equation}\label{eq:zero_bath_stationary_bias}
p_{\alpha,T}(\lambda)-p(\lambda)
=\alpha^2\frac{\Delta_4(\lambda)}{q_2(\lambda)}
+\mathcal O(\alpha^4+\eta(T)).
\end{equation}
All constants in these estimates are independent of $N$.

\begin{proof}[Proof of~\cref{thm:zero_bath_lower_bound_rigor}]
Suppose first that both improvement conditions hold for every $\lambda\in[-1,1]$. Under the assumption $\eta(T)\le C_{\rm tail}\alpha^4$,~\eqref{eq:zero_bath_stationary_bias} gives
\[
\sup_{|\lambda|\le1}|p_{\alpha,T}(\lambda)-p(\lambda)|
=\mathcal O(\alpha^4).
\]
Using~\cref{lem:product_state_error},
\begin{equation}
\|\rho_*-\sigma_1(H)\|_1
\le C\sqrt N\,\alpha^4.
\end{equation}
Furthermore, the proof of~\cref{thm:thermal_zero_bath} shows that there exist constants $C,c>0$ such that
\begin{equation}\label{eq:zero_bath_exact_mixing_bound}
\sup_\rho\|\Phi_\alpha^m(\rho)-\rho_*\|_1
\le CN e^{-c\alpha^2m}.
\end{equation}
Choose $\alpha=\Theta(\epsilon^{1/4}N^{-1/8})$ with a sufficiently small prefactor and choose $T\ge T_0$ so that $\eta(T)\le C_{\rm tail}\alpha^4$. The stretched-exponential tail permits $T=\widetilde{\mathcal O}(1)$. Taking $m=\mathcal O(\alpha^{-2}\log(N/\epsilon))$ makes both errors at most $\epsilon/2$, and hence
\[
T_{\rm total}=2Tm=\widetilde{\mathcal O}(N^{1/4}/\epsilon^{1/2}).
\]
This proves sufficiency for arbitrary free-fermion Hamiltonians in the stated class.

Conversely, fix an energy $\lambda$ at which an improvement condition fails and set $H=\lambda\sum_{j=1}^N n_j$. If the second-order condition fails,~\eqref{eq:zero_bath_leading_stationary_occupation} leaves a nonzero limiting bias. Thus the equilibrium state is $\Omega(1)$ away from the target state. Otherwise $\Delta_4(\lambda)\ne0$. Since $\eta(T)\le C_{\rm tail}\alpha^4$,~\eqref{eq:zero_bath_stationary_bias} becomes
\[
p_{\alpha,T}(\lambda)-p(\lambda)
=\alpha^2\frac{\Delta_4(\lambda)}{q_2(\lambda)}+\mathcal O(\alpha^4).
\]
Thus, for sufficiently small $\alpha_0$ and every $0<\alpha\le\alpha_0$,
\[
|p_{\alpha,T}(\lambda)-p(\lambda)|
\ge\frac{|\Delta_4(\lambda)|}{2q_2(\lambda)}\alpha^2.
\]
For a fixed $p\in(0,1)$ and $r$ in a compact subinterval of $(0,1)$, the elementary binomial comparison is
\begin{equation}\label{eq:zero_bath_binomial_shift}
\|v(r)^{\otimes N}-v(p)^{\otimes N}\|_1
\ge c_p\min\{1,\sqrt N\,|r-p|\}.
\end{equation}
Applying this estimate to~\eqref{eq:zero_bath_stationary_product} yields
\[
\|\rho_*-\sigma_1(H)\|_1\ge c\min\{1,\sqrt N\,\alpha^2\}.
\]
Since $\rho_*$ is an allowed initial state, uniform $\epsilon$-accuracy requires $\alpha^2\lesssim\epsilon/\sqrt N$ for sufficiently small $\epsilon$.

Finally, set $q:=q_{\alpha,T}(\lambda)\in(0,1)$, $p_*:=p_{\alpha,T}(\lambda)$, and
\[
\rho_0^{(b)}:=v(b)^{\otimes N},
\qquad
\rho_m^{(b)}:=\Phi_\alpha^m(\rho_0^{(b)}),
\qquad b\in\{0,1\}.
\]
By~\eqref{eq:zero_bath_population_update},
\[
\rho_m^{(b)}=v(r_m^{(b)})^{\otimes N},
\qquad
r_m^{(b)}=p_*+(b-p_*)(1-q)^m.
\]
Uniform $\epsilon$-accuracy implies
\[
\begin{aligned}
2\epsilon
&\ge\sum_{b=0}^1\|\rho_m^{(b)}-\sigma_1(H)\|_1
\ge\|\rho_m^{(1)}-\rho_m^{(0)}\|_1\ge\left\|\operatorname{Tr}_{2,\ldots,N}
\bigl(\rho_m^{(1)}-\rho_m^{(0)}\bigr)\right\|_1
=\|v(r_m^{(1)})-v(r_m^{(0)})\|_1\\
&=2|r_m^{(1)}-r_m^{(0)}|=2(1-q)^m.
\end{aligned}
\]
Hence, using $q=\Theta(\alpha^2)$,
\[
m\ge\frac{\log(1/\epsilon)}{-\log(1-q)}
\gtrsim\frac1{\alpha^2}\log(1/\epsilon)
\gtrsim\frac{\sqrt N}{\epsilon}\log(1/\epsilon).
\]
This concludes the proof.
\end{proof}

\begin{proof}[Proof of~\cref{lem:zero_bath_lower_bound_rigor}]
Put $q(\nu)=e^{-\sqrt{1+\nu^2}/4}w(\nu/S)$, so that $F_S(\nu)=e^{-\nu/4}q(\nu)$ and the second-order condition holds. Since $q$ is even and positive on $[-1,1]$, the fourth-order condition is equivalent to $D(\nu)=0$ on that interval, where
\[
D(\nu)=e^{\nu/4}\operatorname{Re}J(\nu)-e^{-\nu/4}\operatorname{Re}J(-\nu).
\]
We show that $D'(0)>0$. For $x>0$ and $0\le s\le1/2$, define
\[
M_s(x)=\int_{\mathbb R}q(x/2+y)q(x/2-y)\cosh(sy)\,\mathrm dy,
\qquad k_s=\frac12-s,\quad c_s=\frac12-\frac s2.
\]
A direct Fourier calculation in the ordered time variables gives
\begin{equation}\label{eq:zero_bath_cutoff_derivative}
D'(0)=\frac1{\pi^2}\int_0^\infty\frac1x
\int_0^{1/2}M_s(x)
\bigl[-q'(x)\cosh(c_sx)-k_sq(x)\sinh(c_sx)\bigr]\,\mathrm ds\,\mathrm dx.
\end{equation}
The function $M_s(x)$ is nondecreasing in $s$, whereas $k_s\tanh(c_sx)$ is decreasing. Their opposite monotonicity, under the measure proportional to $\cosh(c_sx)\,\mathrm ds$, implies
\[
\frac{\int_0^{1/2}M_s(x)k_s\sinh(c_sx)\,\mathrm ds}
{\int_0^{1/2}M_s(x)\cosh(c_sx)\,\mathrm ds}
\le\tau(x)
:=\frac{\cosh(x/2)}{2[\sinh(x/2)-\sinh(x/4)]}-\frac2x
\]
whenever the denominator is nonzero. Elementary hyperbolic-function estimates show that $\tau$ is increasing and $0<\tau(x)<5x/48$. Moreover, putting $r=\sqrt{119}/5$ and $m(x)=x/(4\sqrt{1+x^2})$, direct evaluation gives
\[
\begin{aligned}
\tau(5/2)&<21/100<m(r),\\
\tau(11/4)&<223/1000<m(5/2),\\
\tau(3)&<117/500<m(11/4).
\end{aligned}
\]
Since $m(x)\ge5x/48$ for $0<x\le r$, these estimates and monotonicity give $m(x)>\tau(x)$ throughout $(0,3]$. The cutoff assumptions imply $-q'(x)\ge m(x)q(x)$, and $q$ vanishes outside $[-S,S]$ with $S\le3$. Thus~\eqref{eq:zero_bath_cutoff_derivative} has nonnegative integrand after integration in $s$, and it is strictly positive for $0<x<S$. Consequently $D'(0)>0$, contradicting fourth-order preservation on $[-1,1]$.

To obtain~\cref{thm:zero_bath_lower_bound_new_1} with the prescribed choice $S=\|h\|+1$, we set $S=2$ and choose a sufficiently small nonzero energy $\lambda$ at which fourth-order preservation fails. We take $N-1$ modes at energy $\lambda$ and one mode at energy $1$. Then $\|h\|=1$, and the stationary-bias and mixing lower bounds applied to the first $N-1$ modes give $\Omega(\sqrt{N-1}/\epsilon)=\Omega(\sqrt N/\epsilon)$.
\end{proof}

\section{State preparation in the random-bath setting}\label{sec:state_random_bath}

In this section, we consider both thermal state preparation and ground state preparation in the random-bath setting. We first provide a rigorous statement of Theorem~\ref{thm:main_ground}, which gives an upper bound on the complexity of ground state preparation, and then provide a rigorous statement of Theorem~\ref{thm:random_bath_lower_bound}, which establishes a lower bound on the complexity of thermal state preparation.

To simplify the analysis of the upper bound, we consider a setup that is simpler than the one in~\cref{assump:zero_bath}:
\begin{assumption}[Random bath Hamiltonian case]\label{assump:zero_bath_ground}
Let $\{b^\dagger_j,b_j\}^N_{j=1}$ be fermionic creation and annihilation operators on the bath. We set
\[
H_E=\omega\sum^{N}_{i=1}b^\dagger_{i}b_{i},\quad
B_{E,i}=b^\dagger_{i},\quad
\rho_E:\,\text{ground state of }H_E,
\]
and
\[
H_{\rm int}(t)
=f(t)\sum_{i=1}^N
\left(c^\dagger_i\otimes B_{E,i}+c_i\otimes B_{E,i}^\dagger\right)\,.
\]
Here, $\omega$ is sampled uniformly from $[-S,S]$, and $f(t)=\frac{e^{-t^2/(4\sigma^2)}}{\sqrt{\sigma\sqrt{2\pi}}}$.
\end{assumption}

The choice of $f(t)$ in~\cref{assump:zero_bath_ground} is inspired by~\cite{ding2025endtoendefficientquantumthermal,wang2025lindblad}. Although~\cref{assump:zero_bath_ground} is a simplified version of~\cref{assump:zero_bath}, involving only $N$ bath modes, we note that using the same choice of $f(t)$ in the more general setting of~\cref{assump:zero_bath} would lead to the same complexity scaling as in the following theorem:
\begin{thm}\label{thm:random_bath_ground}
Assume $\max_k|\lambda_k|=\mathcal O(1)$ and $\Delta:=\min_k|\lambda_k|>0$. When
\[
\sigma=\Theta(\Delta^{-1}\mathrm{polylog}(N/\epsilon)),\qquad
T=\Theta(\Delta^{-1}\mathrm{polylog}(N/\epsilon)),
\]
$\alpha\sqrt{\sigma}=\Theta(1)$ is sufficiently small, $S>\max_j|\lambda_j|$, and $S=\Theta(1)$, then, for any initial state $\rho$,
\[
\|\Phi^m_{\alpha}\rho-\ketbra{\psi_0}{\psi_0}\|_1\leq \epsilon,\qquad
m=\Theta\left(\frac{1}{\Delta}\mathrm{polylog}(N/\epsilon)\right)\,.
\]
In particular, the ground state $\ket{\psi_0}$ can be prepared to $\epsilon$-accuracy with total Hamiltonian simulation time
\[
T_{\rm total}=\widetilde{\mathcal{O}}(\Delta^{-2})\,.
\]
\end{thm}
We put the proof in~\cref{sec:pf_mixing_time_random_bath_ground}. Here, we suppress additional polylogarithmic factors in $1/\Delta$ in the parameter scalings and complexity bounds for simplicity.

Next, we provide a rigorous statement of Theorem~\ref{thm:random_bath_lower_bound}. We note that the lower bound is established under a more general assumption than~\cref{assump:zero_bath_ground}:
\begin{thm}\label{thm:random_bath_lower_bound_rigor}
For $N\ge2$, set the system bath interaction algorithm with~\cref{assump:zero_bath} and $f(t)\in L^1(\mathbb{R})$ such that
$\mathbb{E}_{\omega}\left(u(\omega;\lambda)\right)>0$ and $\mathbb{E}_{\omega}\left(d(\omega;\lambda)\right)>0$ for all $\lambda\in[-1,1]$.
If there exists $\lambda\in[-1,1]$ such that
\begin{equation}\label{eqn:vanish_variance}
\int^\infty_{-\infty}\left|u(\omega;\lambda)-\exp(-\lambda)d(\omega;\lambda)\right|^2\mathrm{d}\mu(\omega)>0
\end{equation}
then there are positive constants $T_0,\alpha_0,\epsilon_0,\kappa,C$ such that for any $T>T_0$, $\alpha<\alpha_0$, and $\epsilon<\epsilon_0$ satisfying $\int_{|t|>T/2}|f(t)|\mathrm{d}t\leq \kappa \alpha^2$, we have
\[
\sup_{H=\lambda\sum_{i=1}^Nn_i,|\lambda|\leq 1}\inf\left\{\tau\middle|\sup_{\rho}\sup_{m\geq \tau}\|\Phi^m_{\alpha}(\rho)-\sigma_1(H)\|_1\leq \epsilon\right\}\geq \frac{CN}{\epsilon}\,.
\]
In particular, to achieve $\epsilon$-accuracy for thermal state preparation, in the worst case over the stated Hamiltonians, the total Hamiltonian simulation time is $T_{\rm total}=\Omega(N/\epsilon)$.
\end{thm}

As shown in~\cref{sec:pf_random_bath_lower_bound}, although $\mathrm{E}_{\omega}(\Phi_{\alpha})$ may preserve the target thermal state at each single qubit up to arbitrary high order according to~\cite{wang2025lindblad}, it is not guaranteed that the fixed point of $\Phi_{\alpha}$ can approximate direct tensor product of the single-qubit thermal states. In fact, the $\mathcal{O}(\alpha^2)$ term in the fixed point will be a correlated state with a non-negligible deviation from the target thermal state, which gives rise to $\mathcal{O}(N\alpha^2)$ deviation from the target thermal state in the $N$-qubit case. The corresponding variance-vanishing condition, obtained by replacing $>0$ with $=0$ in~\eqref{eqn:vanish_variance}, thus helps to eliminate two-qubit correlations in the $\mathcal{O}(\alpha^2)$ term and prevents these fluctuations from producing an order-$\alpha^2$ deviation from the target thermal state.

\subsection{Proof of \texorpdfstring{\cref{thm:random_bath_ground}}{Theorem~\ref*{thm:random_bath_ground}}}\label{sec:pf_mixing_time_random_bath_ground}

We first show a proposition similar to~\cref{prop:zero_bath} for the random bath case. The proof is similar to that of~\cref{prop:zero_bath}.
\begin{prop}\label{prop:random_bath}
Under~\cref{assump:zero_bath_ground}, there are CPTP maps
$\Phi_{\alpha,j}(\omega):\mathbb C^{2^N\times2^N}\to
\mathbb C^{2^N\times2^N}$, which act on the system mode corresponding to $(d^\dagger_j,d_j)$ and depend on the random variable $\omega$,
\begin{itemize}
\item For $\omega\neq 0$,
$
[\Phi_{\alpha,j}(\omega),\Phi_{\alpha,k}(\omega)]=0
\quad(j\ne k)$.
\item $\Phi_\alpha
=\mathbb{E}_{\omega}\left(\Phi_\alpha\left(\omega\right)\right):=\mathbb{E}_{\omega}\left(\Phi_{\alpha,1}(\omega)\circ\cdots\circ\Phi_{\alpha,N}(\omega)\right).
$
\item For $\omega\neq 0$,
\begin{equation}
\Phi_{\alpha,j}(\omega)[\rho]=\mathrm{Tr}_{B_j}\left(U_{\alpha,j}(\omega)\left(\rho\otimes\rho_{B}(\omega)\right)U_{\alpha,j}^\dagger(\omega)\right)\,,
\end{equation}
where $U_{\alpha,j}(\omega):=\mathcal T\exp\left(
-i\int_{-T}^T H_{\alpha,j,\omega}(s)\,\mathrm ds\right)$, and $e_j = \sum_k U_{k,j} b_k, $
\[
H_{\alpha, j,\omega}(t) =  \lambda_j d_j^\dagger d_j + \omega  e_j^\dagger e_j + \alpha f(t) (d_j^\dagger \otimes e_j^\dagger + d_j \otimes e_j).
\]
\end{itemize}
More specifically,
\begin{equation}
    \Phi_{\alpha, j} = \mathcal{U}_j(T)\circ \Psi_{\alpha,j} \circ \mathcal{U}_j(T)
\end{equation}
with $\mathcal{U}_j(T)[\rho]=\exp(-i\lambda_jd_j^\dagger d_j T)\rho\exp(i\lambda_jd_j^\dagger d_jT)$, and
\begin{equation}
    \begin{aligned}
        \Psi_{\alpha, j}(\rho)& = \sum_{n,m\geq 0}(-1)^n(i\alpha)^{n+m}\int  \int_{-T\leq s_n\leq \cdots \leq s_1 \leq T, -T\leq \tau_m\leq \cdots\leq \tau_1 \leq T}f(s_1)\cdots f(s_n) f(\tau_1)\cdots f(\tau_m)\\
        &\sum_{\substack{\alpha_1,\ldots,\alpha_n\in\{-1,1\}\\\beta_1,\ldots,\beta_m\in\{-1,1\}}}
e^{-i(\omega+\lambda_j)\sum_{q=1}^m\beta_q\tau_q
+i(\omega+\lambda_j)\sum_{p=1}^{n}\alpha_p s_p}\mathrm{d}s_1 \cdots \mathrm{d}s_n \mathrm{d}\tau_1\cdots \mathrm{d}\tau_m \\
        &d_{\alpha_1,j}\cdots d_{\alpha_n,j} \rho  d_{-\beta_m,j} \cdots d_{-\beta_1,j}\mathrm{Tr}_{E_j}(e_{\alpha_1,j}\cdots e_{\alpha_n,j}\rho_B(\omega)  e_{-\beta_m,j} \cdots e_{-\beta_1,j} ) \,.\\
    \end{aligned}
\end{equation}
with $e_{\eta,j} = \begin{cases}
    e_j^\dagger & \eta = 1\\ e_j & \eta= -1
\end{cases}$, $d_{\eta,j} = \begin{cases}
    d_j^\dagger & \eta = 1\\ d_j & \eta= -1
\end{cases}$.
\end{prop}
\begin{proof}[Proof of~\cref{prop:random_bath}]
    Under~\cref{assump:zero_bath_ground}, the total Hamiltonian is
\begin{equation}
\begin{aligned}
    H_\alpha(t) &= H + H_E + \alpha H_{\rm int}(t)\\
    & = \sum_{k = 1}^N \lambda_k d_k^\dagger d_k + \omega \sum_{k = 1}^N b_k^\dagger b_k + \alpha f(t)\sum_{k=1}^N (c_k^\dagger \otimes B_{E,k} + c_k \otimes B_{E,k}^\dagger)\,.
\end{aligned}
\end{equation}
After applying the unitary transformation, the system Hamiltonian is diagonalized in terms of $d_k, d_k^\dagger$, where the bath Hamiltonian keeps the diagonal form with $e_k, e_k^\dagger$. Here
\begin{equation}
    e_j = \sum_k U_{k,j} b_k, \quad e_j^\dagger = \sum_k U_{kj}^* b_k^\dagger\,.
\end{equation}
Therefore the interaction term is decomposed as a summation of one mode terms so that the total Hamiltonian can be rewritten as
\begin{equation}
    \begin{aligned}
        H_\alpha(t) = \sum_{k = 1}^N \lambda_k d_k^\dagger d_k + \omega \sum_{k = 1}^N e_k^\dagger e_k + \alpha f(t)\sum_{k=1}^N (d_k^\dagger \otimes e_k^\dagger + d_k \otimes e_k)= \sum_k H_{\alpha, k,\omega}(t)\,.\\
    \end{aligned}
\end{equation}
Here
\begin{equation}
    H_{\alpha, j,\omega}(t) =  \lambda_j d_j^\dagger d_j + \omega  e_j^\dagger e_j + \alpha f(t) (d_j^\dagger \otimes e_j^\dagger + d_j \otimes e_j).
\end{equation}
Those single mode Hamiltonian commutes with each other $[H_{\alpha, k,\omega}(t), H_{\alpha, j,\omega}(s)] = 0$ when $k \neq j$. Denote its unitary evolution operator as $U_{\alpha,j}$, the total evolution operator
\begin{equation}
\begin{aligned}
    U^\alpha(T) &= \mathcal{T}\left(\exp\left(-i\int_{-T}^T H_{\alpha}(s)\mathrm{d}s\right)\right)
     = \prod_k \mathcal{T}\left(\exp\left(-i\int_{-T}^T H_{\alpha,k,\omega}(s)\mathrm{d}s\right)\right) = \prod_kU_{\alpha, k }(T)
\end{aligned}
\end{equation}
The induced quantum channel becomes
\begin{equation}
    \Phi_\alpha(\rho) = \mathbb{E}_\omega[\Phi_{\alpha, 1}(\omega) \circ \Phi_{\alpha, 2}(\omega)\circ \cdots \circ \Phi_{\alpha, N}(\omega) (\rho)].
\end{equation}
Let $S_{\eta,j} = \begin{cases}
    d_j^\dagger & \eta = +1\\ d_j & \eta = -1\\
\end{cases}$, $B_{\eta,j} = \begin{cases}
    e_j^\dagger & \eta = +1 \\ e_j & \eta = -1\\
\end{cases}. $ From the Dyson series expansion, we have
\begin{equation}
    \Phi_{\alpha, j}(\rho) = \mathcal{U}_{j}(T) \circ \Psi_{\alpha, j}\circ  \mathcal{U}_{j}(T) (\rho).
\end{equation}
where
\begin{equation}
    \begin{aligned}
        \Psi_{\alpha, j}(\rho)& = \sum_{n,m\geq 0}(-1)^n(i\alpha)^{n+m}\int  \int_{-T\leq s_n\leq \cdots \leq s_1 \leq T, -T\leq \tau_m\leq \cdots\leq \tau_1 \leq T}\mathrm{d}s_1 \cdots \mathrm{d}s_n \mathrm{d}\tau_1\cdots \mathrm{d}\tau_m \\
        &\sum_{\alpha_i,\beta_j\in\{-1,1\}}S_{\alpha_1,j}(s_1)\cdots S_{\alpha_n,j}(s_n) \rho  S_{\beta_m,j}^\dag(\tau_m) \cdots S_{\beta_1,j}^\dag (\tau_1)\\
        & f(s_1)\cdots f(s_n) f(\tau_1)\cdots f(\tau_m)\mathrm{Tr}_{E_j}(B_{\alpha_1,j}(s_1)\cdots B_{\alpha_n,j}(s_n)\rho_B(\omega)  B_{\beta_m,j}^\dag(\tau_m) \cdots B_{\beta_1,j}^\dag (\tau_1))\,, \\
    \end{aligned}
\end{equation}
The Heisenbergy picture of the system and bath coupling operator can be characterized as
\begin{equation}
    S_{\eta, j}(t) = e^{i\eta\lambda_j t}S_{\eta, j}, \quad B_{\eta, j}(t) = e^{i\eta\omega t}B_{\eta,j}\,.
\end{equation}
Therefore, the expression can be further simplified to
\begin{equation}
    \begin{aligned}
        \Psi_{\alpha, j}(\rho)& = \sum_{n,m\geq 0}(-1)^n(i\alpha)^{n+m}\int  \int_{-T\leq s_n\leq \cdots \leq s_1 \leq T, -T\leq \tau_m\leq \cdots\leq \tau_1 \leq T}f(s_1)\cdots f(s_n) f(\tau_1)\cdots f(\tau_m)\\
        &\sum_{\substack{\alpha_1,\ldots,\alpha_n\in\{-1,1\}\\\beta_1,\ldots,\beta_m\in\{-1,1\}}}
e^{-i(\omega+\lambda_j)\sum_{q=1}^m\beta_q\tau_q
+i(\omega+\lambda_j)\sum_{p=1}^{n}\alpha_p s_p}\mathrm{d}s_1 \cdots \mathrm{d}s_n \mathrm{d}\tau_1\cdots \mathrm{d}\tau_m \\
        &S_{\alpha_1,j}\cdots S_{\alpha_n,j} \rho  S_{\beta_m,j}^\dag \cdots S_{\beta_1,j}^\dag\mathrm{Tr}_{E_j}(B_{\alpha_1,j}\cdots B_{\alpha_n,j}\rho_B(\omega)  B_{\beta_m,j}^\dag \cdots B_{\beta_1,j}^\dag ) \,.\\
    \end{aligned}
\end{equation}
\end{proof}

Specifically, the initial bath state is
\begin{equation}
    \rho_B(\omega)=\left\{\begin{aligned}\ket{0}\bra{0}\ \text{if}\ \omega>0\\
\ket{1}\bra{1}\ \text{if}\ \omega<0
\end{aligned}\right.\,.
\end{equation}
Define the number operator $n_j=d_j^\dagger d_j$ for the $j$-th system mode. We define the approximation of the induced channel as
\begin{equation}
    \widetilde{\Phi}_{\alpha,j} =\mathcal{U}_j(T)\circ \widetilde{\Psi}_{\alpha,j}\circ \mathcal{U}_j(T)
\end{equation}
where $\widetilde{\Psi}_j$ is the infinite time window approximation of the middle step mapping $\Psi_j$.
Then, we have the following lemma:
\begin{lem}\label{lem:random_bath_ground}
Under the settings and results of~\cref{prop:random_bath},
\begin{itemize}
  \item For $\omega\neq 0$, $\Phi^\dagger_{\alpha,j}(\omega)[n_k]=\widetilde{\Phi}^\dagger_{\alpha,j}(\omega)[n_k]=n_k$ for $k\neq j$.
  \item For $\omega\neq 0$, $\left\|\widetilde{\Phi}_{\alpha,j}(\omega)-\Phi_{\alpha,j}(\omega)\right\|_{1\leftrightarrow 1}=\mathcal{O}\left(\frac{\alpha^2\sigma^2}{T}\exp(-T^2/(4\sigma^2))\right)$.
  \item When $\sigma=\Omega(\Delta^{-1})$ large enough, $\alpha\sqrt{\sigma}=\mathcal{O}(1)$ sufficiently small, $S>\max_j|\lambda_j|$, and $S=\Theta(1)$, there exists uniform constants $C,C_0>0$ such that
  \begin{itemize}
  \item If $\lambda_j>0$, $\left(\mathbb{E}_{\omega}\left(\widetilde{\Phi}_{\alpha,j}(\omega)\right)\right)^\dagger[n_j]\preceq \left(1-C\alpha^2\right)n_j+\sum^\infty_{i=1}\frac{\left(C_0\alpha\sqrt{\sigma}\right)^{2i-1}}{(2i-1)!}\exp\left(-\frac{\sigma^2\Delta^2}{(2i-1)}\right) I$.
  \item If $\lambda_j<0$, $\left(\mathbb{E}_{\omega}\left(\widetilde{\Phi}_{\alpha,j}(\omega)\right)\right)^\dagger[I-n_j]\preceq \left(1-C\alpha^2\right)(I-n_j)+\sum^\infty_{i=1}\frac{\left(C_0\alpha\sqrt{\sigma}\right)^{2i-1}}{(2i-1)!}\exp\left(-\frac{\sigma^2\Delta^2}{(2i-1)}\right) I$\,.
  \end{itemize}
\end{itemize}
\end{lem}
The proof of~\cref{lem:random_bath_ground} is similar to that of~\cite[Theorem S14]{wang2025lindblad}. For completeness, we provide the proof in~\cref{sec:pf_random_bath_ground_lem}.

Now, we are ready to prove~\cref{thm:random_bath_ground}.
\begin{proof}[Proof of~\cref{thm:random_bath_ground}]
In the proof, we consider the decay of the total number operator $\hat{N}=\sum_{j=1}^N \left(n_j\textbf{1}_{\lambda_j>0}+(I-n_j)\textbf{1}_{\lambda_j<0}\right)$. First, according to the commuting property in~\cref{lem:random_bath_ground}, for each $j$ and $m>0$, we have
\[
\left(\Phi^m_{\alpha}\right)^\dagger\left(n_j\right)=\left(\left(\mathbb{E}_{\omega}\left(\Phi_{\alpha,j}(\omega)\right)\right)^m\right)^\dagger[n_j]=\left(\left(\mathbb{E}_{\omega}\left(\widetilde{\Phi}_{\alpha,j}(\omega)\right)\right)^m\right)^\dagger[n_j]+\mathcal{O}\left(m\alpha^2\sigma^2\exp(-T^2/(4\sigma^2))/T\right)\,.
\]
This implies that
\[
\begin{aligned}
&\mathrm{Tr}\left(\Phi^m_{\alpha}[\rho]\hat{N}\right)=\sum^N_{j=1}\left[\mathrm{Tr}\left(\left(\mathbb{E}_{\omega}\left(\Phi_{\alpha,j}(\omega)\right)\right)^m[\rho]n_j\right)\textbf{1}_{\lambda_j>0}+\mathrm{Tr}\left(\left(\mathbb{E}_{\omega}\left(\Phi_{\alpha,j}(\omega)\right)\right)^m[\rho](I-n_j)\right)\textbf{1}_{\lambda_j<0}\right]\\
=&\sum^N_{j=1}\left[\mathrm{Tr}\left(\left(\mathbb{E}_{\omega}\left(\widetilde{\Phi}_{\alpha,j}(\omega)\right)\right)^m[\rho]n_j\right)\textbf{1}_{\lambda_j>0}+\mathrm{Tr}\left(\left(\mathbb{E}_{\omega}\left(\widetilde{\Phi}_{\alpha,j}(\omega)\right)\right)^m[\rho]\left(I-n_j\right)\right)\textbf{1}_{\lambda_j<0}\right]+\mathcal{O}\left(mN\alpha^2\sigma^2\exp(-T^2/(4\sigma^2))/T\right)\\
\leq &\left(1-C\alpha^2\right)^m\mathrm{Tr}\left(\rho\hat{N}\right)+\sum^N_{j=1}\sum^\infty_{i=1}m\frac{\left(C_0\alpha\sqrt{\sigma}\right)^{2i-1}}{(2i-1)!}\exp\left(-\frac{\sigma^2\Delta^2}{(2i-1)}\right)+\mathcal{O}\left(mN\alpha^2\sigma^2\exp(-T^2/(4\sigma^2))/T\right)\\
\end{aligned}\,.
\]
By Fuchs-van de Graaf inequality, for any density matrix $\rho$,
\[
\|\rho-\ketbra{\psi_0}{\psi_0}\|_1\leq 2\sqrt{1-\bra{\psi_0}\rho\ket{\psi_0}}\leq 2\sqrt{\mathrm{Tr}(\rho\hat{N})}\,.
\]
Combining this with $\mathrm{Tr}\left(\rho\hat{N}\right)\leq N$, we have
\[
\|\Phi^m_{\alpha}\rho-\ketbra{\psi_0}{\psi_0}\|_1\leq 2\left(1-C\alpha^2\right)^{m/2}\sqrt{N}+2\sqrt{\sum^N_{j=1}\sum^\infty_{i=1}m\frac{\left(C_0\alpha\sqrt{\sigma}\right)^{2i-1}}{(2i-1)!}\exp\left(-\frac{\sigma^2\Delta^2}{(2i-1)}\right)+\mathcal{O}\left(mN\alpha^2\sigma^2\exp(-T^2/(4\sigma^2))/T\right)}
\]
Plugging in the choices of $\sigma$, $T$, and $\alpha$, we obtain the result.
\end{proof}

\subsection{Proof of~\texorpdfstring{\cref{thm:random_bath_lower_bound_rigor}}{Theorem~\ref*{thm:random_bath_lower_bound_rigor}}}\label{sec:pf_random_bath_lower_bound}

We establish the lower bound for $H=\lambda\sum_{j=1}^N n_j$. Fix $\omega$. We first consider one system mode, with $d=\ket{0}\bra{1}$, $n=d^\dagger d$, and two bath modes $b_+$ and $b_-$ as in~\cref{assump:zero_bath}. The reduced channel is
\begin{equation}\label{eq:random_bath_single_mode_channel}
\Phi^{(1)}_{\alpha,\omega}(\rho)
=\operatorname{Tr}_B\!\left[U(\rho\otimes\sigma_B)U^\dagger\right],
\end{equation}
where $\sigma_B=\sigma_b^{\otimes2}$,
\[
\sigma_b=(1-p_{\rm b}(\omega))\ket{0}\bra{0}
+p_{\rm b}(\omega)\ket{1}\bra{1}\,.
\]
and $U=\mathcal T\exp(-i\int_{-T}^T H^{(1)}_{\alpha,\omega}(t)\,\mathrm dt)$, with
\begin{equation}
\begin{aligned}
H^{(1)}_{\alpha,\omega}(t)
&=\lambda n+\omega(n_{b,+}+n_{b,-})+\alpha\left(d^\dagger\otimes F(t)+d\otimes F(t)^\dagger\right),\\
F(t)&=f(t)b_+^\dagger+\overline{f(t)}b_-,
\end{aligned}
\end{equation}
where $n_{b,\pm}=b_\pm^\dagger b_\pm$. In particular, $d^\dagger\ket{0}=\ket{1}$ and $d\ket{1}=\ket{0}$, so the interaction sends $\ket{0}\otimes\ket{\chi}$ to $\ket{1}\otimes F(t)\ket{\chi}$ and $\ket{1}\otimes\ket{\chi}$ to $\ket{0}\otimes F(t)^\dagger\ket{\chi}$ for any bath vector $\ket{\chi}$.

Let $U_{rs}=(\bra{r}\otimes I_B)U(\ket{s}\otimes I_B)$, $r,s\in\{0,1\}$, denote the bath-operator blocks of $U$. The joint evolution therefore has the exact form
\begin{equation}
\begin{aligned}
U(\ket{0}\otimes\ket{\chi})
&=\ket{0}\otimes U_{00}\ket{\chi}
+\ket{1}\otimes U_{10}\ket{\chi},\\
U(\ket{1}\otimes\ket{\chi})
&=\ket{0}\otimes U_{01}\ket{\chi}
+\ket{1}\otimes U_{11}\ket{\chi}.
\end{aligned}
\end{equation}
Taking the bath trace gives
\begin{equation}
\Phi^{(1)}_{\alpha,\omega}({\ket{s}\bra{s}})
=\sum_{r,r'=0}^1
\operatorname{Tr}_B(U_{rs}\sigma_BU_{r's}^\dagger)
\ket{r}\bra{r'}.
\end{equation}
The conserved parity $(-1)^{n+n_{b,+}+n_{b,-}}$ and diagonal bath state imply $\operatorname{Tr}_B(U_{rs}\sigma_BU_{r's}^\dagger)=0$ for $r\ne r'$. Unitarity also implies $\sum_rU_{rs}^\dagger U_{rs}=I_B$. Consequently, the channel acts exactly as
\begin{equation}\label{eq:random_bath_basis_state_action}
\begin{aligned}
\Phi^{(1)}_{\alpha,\omega}({\ket{0}\bra{0}})
&=(1-a(\omega)){\ket{0}\bra{0}}+a(\omega){\ket{1}\bra{1}},\\
\Phi^{(1)}_{\alpha,\omega}({\ket{1}\bra{1}})
&=z(\omega){\ket{0}\bra{0}}+(1-z(\omega)){\ket{1}\bra{1}},
\end{aligned}
\end{equation}
where
\begin{equation}\label{eq:random_bath_transition_probabilities}
\begin{aligned}
a(\omega)&=\operatorname{Tr}_B(U_{10}\sigma_BU_{10}^\dagger)
=\bra{1}\Phi^{(1)}_{\alpha,\omega}({\ket{0}\bra{0}})\ket{1},\\
z(\omega)&=\operatorname{Tr}_B(U_{01}\sigma_BU_{01}^\dagger)
=\bra{0}\Phi^{(1)}_{\alpha,\omega}({\ket{1}\bra{1}})\ket{0}.
\end{aligned}
\end{equation}
In particular, $0\le a(\omega),z(\omega)\le1$, and occupation-diagonal inputs remain diagonal without any approximation in $\alpha$. Set $F_T(\nu)=\int_{-T}^T f(t)e^{i\nu t}\,\mathrm dt$ and define $r=p_{\rm b}(\omega)$. It is straightforward to check that
\begin{equation}\label{eq:random_bath_leading_population_rates}
\begin{aligned}
a(\omega)=\alpha^2a_2(\omega)+\mathcal O(\alpha^4),\quad a_2(\omega)=(1-r)\left|F_T(\lambda+\omega)\right|^2+r\left|F_T(\omega-\lambda)\right|^2,\\
z(\omega)=\alpha^2z_2(\omega)+\mathcal O(\alpha^4),\quad z_2(\omega)=(1-r)\left|F_T(\omega-\lambda)\right|^2+r\left|F_T(\lambda+\omega)\right|^2.
\end{aligned}
\end{equation}
The remainders in~\eqref{eq:random_bath_leading_population_rates} are uniform in $\omega$ and $T$, since the interaction-picture coupling has norm at most $2|f(t)|$ and $f\in L^1(\mathbb R)$.
The coefficients in~\cref{eq:random_bath_transition_probabilities} are Born probabilities for occupation measurements. If $X_t\in\{0,1\}$ denotes the system occupation, then
\begin{equation}
\begin{aligned}
\mathbb P(X_{t+1}=1\mid X_t=0,\omega)&=a(\omega),\\
\mathbb P(X_{t+1}=0\mid X_t=1,\omega)&=z(\omega).
\end{aligned}
\end{equation}
On diagonal states, the action of the quantum channel is described by the scalar occupation-probability map $\mathcal T_\omega$: For $p\in[0,1]$, let $v(p)=(1-p)\ket{0}\bra{0}+p\ket{1}\bra{1}\in\mathcal D(\mathbb C^2)$ denote the diagonal one-qubit density operator with occupation probability $p$. By linearity and~\cref{eq:random_bath_basis_state_action},
\begin{equation}
\begin{aligned}
\Phi^{(1)}_{\alpha,\omega}(v(p))=[(1-p)(1-a(\omega))+pz(\omega)]\ket{0}\bra{0}+[(1-p)a(\omega)+p(1-z(\omega))]\ket{1}\bra{1}=v(\mathcal T_\omega(p)),
\end{aligned}
\end{equation}
where
\begin{equation}\label{eq:random_bath_probability_update}
\begin{aligned}
\mathcal T_\omega(p)=(1-p)a(\omega)+p(1-z(\omega))=a(\omega)+(1-q(\omega))p,\quad q(\omega)&=a(\omega)+z(\omega).
\end{aligned}
\end{equation}
For independent frequencies $\omega_t\sim\mu$ sampled at successive iterations, the occupation probability evolves as the Markov chain
\[
P_{t+1}=\mathcal T_{\omega_t}(P_t),\qquad P_0=p\in[0,1].
\]
Here $P_t$ is the occupation probability conditioned on the sampled frequency history. Averaging a single update gives
\[
\mathbb E_\omega\!\left[\Phi^{(1)}_{\alpha,\omega}(v(p))\right]
=v\!\left(\mathbb E_\omega[\mathcal T_\omega(p)]\right).
\]
For sufficiently small $\alpha$ and sufficiently large $T$ as in the theorem, $0\le q(\omega)<1$ uniformly and $\mathbb E_\omega[q(\omega)]>0$, so the random affine maps contract on average. Let $\nu$ denote the stationary distribution of the chain $P_t$ on $[0,1]$. By definition, if $P_\infty\sim\nu$ and a fresh $\omega\sim\mu$ are independent, then
 \begin{equation}
     \mathcal{T}_\omega(P_\infty)\sim \nu.
 \end{equation}
The stationary distribution can also be characterized as the limiting distribution of $P_t$ when $t\rightarrow\infty$, or through the backward limit $ \lim_{\tau\rightarrow \infty} \mathcal{T}_{\omega_{-1}}\circ \mathcal{T}_{\omega_{-2}}\circ \cdots \circ \mathcal{T}_{\omega_{-\tau}}(p)$ for any initial value $p\in[0,1]$. We finally calculate the expectation and variance of $P_\infty$:
\begin{itemize}
    \item Expectation: $\mu_{P_\infty}=\mathbb{E}[P_\infty] = \frac{\mathbb{E}_{\omega} [a(\omega)]}{\mathbb{E}_{\omega}[q(\omega)]}.$
    \item Variance: $\mathrm{Var}[P_\infty] = \frac{\mathbb{E}_{\omega}[(a(\omega) -\mu_{P_\infty}q(\omega))^2]}{\mathbb{E}_{\omega} [2q(\omega)- q^2(\omega)]}$.
\end{itemize}
This finishes the discussion of the single-mode case.

For $N$ modes, the joint Hamiltonian is a sum of commuting terms, each acting on one system mode and its two bath modes. The conserved parities $(-1)^{n_j+n_{b,j,+}+n_{b,j,-}}$ imply diagonal preservation. Since $\lambda_j=\lambda$, all modes have the same occupation-probability update $\mathcal T_\omega$. At fixed $\omega$, the evolution of a product diagonal state therefore satisfies
\begin{equation}\label{eq:random_bath_product_population_update}
\begin{aligned}
\Phi_{\alpha,\omega}(v(p)^{\otimes N})=\bigotimes_{j=1}^N\Phi^{(1)}_{\alpha,\omega}(v(p))=v(\mathcal T_\omega(p))^{\otimes N}.
\end{aligned}
\end{equation}
Starting from $v(p)^{\otimes N}$, the state conditioned on the frequency history is thus $v(P_t)^{\otimes N}$, whereas the averaged state is
\[
\Phi_\alpha^t(v(p)^{\otimes N})
=\mathbb E[v(P_t)^{\otimes N}].
\]
The same frequency is shared by all modes at each iteration. Consequently, averaging over frequency histories can create classical correlations even though the state is a product conditional on the frequency history.

Define the stationary density operator
\begin{equation}
\rho_*=\mathbb E_{P_\infty\sim\nu}[v(P_\infty)^{\otimes N}].
\end{equation}
Using~\cref{eq:random_bath_product_population_update} and averaging over both $P_\infty$ and an independent fresh frequency gives
\begin{equation}
\begin{aligned}
\Phi_\alpha(\rho_*)=\mathbb E_{P_\infty,\omega}
\!\left[v(\mathcal T_\omega(P_\infty))^{\otimes N}\right]=\mathbb E_{P_\infty\sim\nu}[v(P_\infty)^{\otimes N}]=\rho_*.
\end{aligned}
\end{equation}

Let the binomial function  $b_{N,k}(p) = \binom{N}{k} p^k (1-p)^{N-k},$ and $p_\beta(\lambda) = \frac{1}{1+ e^{\beta \lambda}}$ denote the occupation probability of the state $\ketbra{1}$ in the Gibbs state. The difference between the fixed point $\rho_*$ and the Gibbs state can be characterized as
\begin{equation}
    \begin{aligned}
        \|\rho_* - \sigma_\beta\|_1 &= \sum_{k = 0}^N |\mathbb{E}[b_{N,k}(P_\infty)] - b_{N,k}(p_\beta)|\\
    \end{aligned}
\end{equation}

We introduce the following lemma to bound the difference between the corresponding binomial probabilities.
\begin{lem}\label{lem:bino_mix amplification}
Given a random variable $ P \in [0,1]$, and a fixed value $p \in(0,1)$. Suppose that there exists $C>0$ such that
\begin{equation}
    \mathbb{E}|P - \mathbb{E}P|^4 \leq C \mathrm{Var}(P)^2\,,
\end{equation}
then there exists a constant $c>0$, depending on $p$ and $C$, such that for every $N\geq 2$
\begin{equation}
    \sum_{k = 0}^N |\mathbb{E}[b_{N,k}(P)] - b_{N,k}(p)|\geq c \min \{1, N (\mathrm{Var}(P) + (\mathbb{E}(P) - p)^2)\}\,.
\end{equation}
\end{lem}

\begin{proof}[Proof of \cref{lem:bino_mix amplification}]
    We may assume $C\ge1$. We consider the function for $0<\theta\leq 1/\sqrt{N}$,
    \begin{equation}
        h(p) \coloneqq \sum_{k = 0}^N e^{ik\theta}b_{N,k}(p) = (1+(e^{i\theta} - 1)p)^N\,.
    \end{equation}
    It corresponds to the generating function of $P$ such that
    \begin{equation}\label{eq:first}
        |\mathbb{E}[h(P)] - h(p)| = \left|\sum_{k = 0}^N e^{ik\theta} \left(\mathbb{E}[b_{N,k}(P)] - b_{N,k}\left(p\right)\right)\right|\leq \sum_{k = 0}^N  \left|\mathbb{E}[b_{N,k}(P)] - b_{N,k}\left(p\right)\right|\,.
    \end{equation}
    We therefore derive a lower bound using Taylor expansion of the function $h$. Notice that
    \begin{equation}
        h(q) = h(p) + h'(p)(q-p) + \frac{h''(p)}{2!}(q-p)^2 +\frac{1}{2}\int_p^q (q-s)^2 h'''(s)\mathrm{d}s 
    \end{equation}
    where
    \begin{equation}
        |h'''(s)| = |N(N-1)(N-2) (1+(e^{i\theta} - 1)s)^{N-3}(e^{i\theta} - 1)^3|\leq N^3\theta^3\,.
    \end{equation}
    We decompose the difference into the leading term and the remainder
    $$\mathbb{E}[h(P)] - h(p) = M + R$$
    \begin{itemize}
        \item The contribution from the first and second-order terms satisfies
        \begin{equation}
        \begin{aligned}
            |M| &= |h'(p)(\mathbb{E}[P]-p) +  \frac{h''(p)}{2!}\mathbb{E}[(P-p)^2]| = \left|h'(p)\left((\mathbb{E}[P]-p) +  \frac{h''(p)}{2! h'(p)}\mathbb{E}[(P-p)^2]\right)\right|\\
            &\geq \frac{| h'(p)|}{2} \mathrm{Im}\left(\frac{h''(p)}{h'(p)}\right)\mathbb{E}[(P-p)^2]\geq \frac{N^2\theta^2}{32} \mathbb{E}[(P-p)^2]\,.
        \end{aligned}
        \end{equation}
        \item For the remainder term, we have
        \begin{equation}
            |R|\leq \frac{N^3\theta^3}{6} \mathbb{E}[|P-p|^3] 
        \end{equation}
        Using $(a+b)^4 \le 8(a^4 + b^4)$ together with the assumption, we obtain
        \begin{equation}
            \mathbb{E}[(P - p)^4] \leq 8 \left( C\mathrm{Var}(P)^2 + (\mathbb{E}P -p)^4\right)\leq 8C \left(\mathbb{E}[\left(P-p\right)^2]\right)^2,
        \end{equation}
        By Cauchy--Schwarz, $\mathbb E[|P-p|^3]\leq\sqrt{\mathbb E[(P-p)^2]\mathbb E[(P-p)^4]}\leq\sqrt{8C}\,\mathbb E[(P-p)^2]^{3/2}$, and therefore
        \begin{equation}
            |R|\leq \frac{\sqrt{8C}\,N^3\theta^3}{6}\mathbb{E}[(P-p)^2]^{3/2}\,.
        \end{equation}
    \end{itemize}
    Set $S=\mathbb E[(P-p)^2]$. If $S=0$, the claim is immediate. Otherwise, combining these bounds gives
    \begin{equation}
        \left|\mathbb{E}[h(P)] - h(p)\right|\geq N^2\theta^2 S\left(\frac1{32}-\frac{\sqrt{8C}}6N\theta\sqrt S\right).
    \end{equation}
    Choose $\theta=c_C\min\{N^{-1/2},(N\sqrt S)^{-1}\}$, where $0<c_C\leq\min\{1,3/(32\sqrt{8C})\}$. The factor in parentheses is then at least $1/64$, so
    \begin{equation}\label{eq:last}
        \left|\mathbb{E}[h(P)] - h(p)\right|\gtrsim \min\{1, N \mathbb{E}[(P-p)^2]\}\,.
    \end{equation}
    Finally, since $\mathbb{E}[(P-p)^2] = \mathrm{Var}(P) + (\mathbb{E}(P) - p)^2$, the result follows by combining~\cref{eq:first} and~\cref{eq:last}.
\end{proof}

Now, we are ready to prove Theorem~\ref{thm:random_bath_lower_bound_rigor}:
\begin{proof}

Fix $\lambda\in[-1,1]$ satisfying~\eqref{eqn:vanish_variance}, and set $H=\lambda\sum_{j=1}^Nn_j$, $p_1=(1+e^\lambda)^{-1}$, $Q(\omega)=u(\omega;\lambda)+d(\omega;\lambda)$, $\overline Q=\mathbb E_\omega[Q(\omega)]>0$, and $p_0=\mathbb E_\omega[u(\omega;\lambda)]/\overline Q$. The tail assumption gives
\[
\sup_\nu|F_T(\nu)-\widehat f(\nu)|
\leq\int_{|t|>T}|f(t)|\,\mathrm dt\leq\kappa\alpha^2.
\]
Together with~\eqref{eq:random_bath_leading_population_rates}, this yields, uniformly in $\omega$ and admissible $T$,
\begin{equation}
\begin{aligned}
a(\omega)&=\alpha^2u(\omega;\lambda)+\mathcal O(\alpha^4),\quad z(\omega)=\alpha^2d(\omega;\lambda)+\mathcal O(\alpha^4),\\
\overline q:=\mathbb E_\omega[q(\omega)]
&=\alpha^2\overline Q+\mathcal O(\alpha^4)=\Theta(\alpha^2).
\end{aligned}
\end{equation}
In particular, $\mu_{P_\infty}=\mathbb E_\omega[a(\omega)]/\overline q=p_0+\mathcal O(\alpha^2)$.

If $p_0\ne p_1$, then for sufficiently small $\alpha$,
\[
\|\rho_*-\sigma_1(H)\|_1
\geq2|\mu_{P_\infty}-p_1|\geq|p_0-p_1|>0,
\]
by measuring one system occupation. Since $\rho_*$ is stationary, the theorem's accuracy requirement is impossible for sufficiently small $\epsilon$, and the claimed lower bound holds. It remains to consider $p_0=p_1$, which is equivalent to~\eqref{eqn:random_bath_second_order_thermal_condition}. In this case,
\begin{equation}\label{eqn:random_bath_expectation_thermal_condition}
\mu_{P_\infty}-p_1=\mathcal O(\alpha^2).
\end{equation}
\begin{equation}\label{eqn:random_bath_variance_thermal_condition}
\operatorname{Var}(P_\infty)
=\alpha^2\frac{\mathbb E_\omega[(u(\omega;\lambda)-p_1Q(\omega))^2]}{2\overline Q}
+\mathcal O(\alpha^4)=\Theta(\alpha^2).
\end{equation}
Indeed, $u-p_1Q=(u-e^{-\lambda}d)/(1+e^{-\lambda})$, so the leading coefficient is strictly positive by~\eqref{eqn:vanish_variance}.

To verify the fourth-moment hypothesis of~\cref{lem:bino_mix amplification}, write $X=P_\infty-\mu_{P_\infty}$, $A=1-q(\omega)$, and $B=a(\omega)-\mu_{P_\infty}q(\omega)$. Stationarity and independence of $X$ and the fresh frequency give $AX+B\overset{d}=X$, $\mathbb E_\omega[B]=0$, and
\begin{equation}
\begin{aligned}
(1-\mathbb E_\omega[A^4])\mathbb E[X^4]=\mathbb E_\omega[B^4]
+6\mathbb E_\omega[A^2B^2]\mathbb E[X^2]+4\mathbb E_\omega[A^3B]\mathbb E[X^3].
\end{aligned}
\end{equation}
Here $\|q\|_\infty,\|B\|_\infty=\mathcal O(\alpha^2)$, $\mathbb E[X^2]=\Theta(\alpha^2)$, and $|\mathbb E[X^3]|\leq\mathbb E[X^2]$ because $|X|\leq1$. Moreover, $\mathbb E_\omega[A^3B]=\mathbb E_\omega[(A^3-1)B]=\mathcal O(\alpha^4)$ and $1-\mathbb E_\omega[A^4]=\Theta(\alpha^2)$. Consequently,
\[
\mathbb E[X^4]=\mathcal O(\alpha^4)
=\mathcal O\!\left(\operatorname{Var}(P_\infty)^2\right).
\]
Applying~\cref{lem:bino_mix amplification} with $p=p_1$ and $N\ge2$ therefore yields
\begin{equation}
\begin{aligned}
\|\rho_*-\sigma_1(H)\|_1\geq c\min\{1,N(\operatorname{Var}(P_\infty)+(\mu_{P_\infty}-p_1)^2)\}\geq c'\min\{1,N\alpha^2\}.
\end{aligned}
\end{equation}
Since $\rho_*$ is an allowed initial state, the theorem's uniform accuracy requirement forces $\alpha^2\lesssim\epsilon/N$ for sufficiently small $\epsilon$.

Finally, applying the averaged occupation update to the initial states $v(0)^{\otimes N}$ and $v(1)^{\otimes N}$ gives
\[
\operatorname{Tr}\!\left[n_1\Phi_\alpha^m(v(1)^{\otimes N})\right]
-\operatorname{Tr}\!\left[n_1\Phi_\alpha^m(v(0)^{\otimes N})\right]
=(1-\overline q)^m.
\]
If both outputs are within trace norm $\epsilon$ of $\sigma_1(H)$, measuring the first occupation implies $(1-\overline q)^m\leq\epsilon$. Since $\overline q=\Theta(\alpha^2)$,
\[
m\geq\frac{\log(1/\epsilon)}{-\log(1-\overline q)}
\gtrsim\frac1{\alpha^2}\log(1/\epsilon)
\gtrsim\frac{N}{\epsilon}\log(1/\epsilon).
\]
This proves the stated worst-case iteration lower bound. Since $T>T_0>0$, the total Hamiltonian simulation time is also $\Omega(N/\epsilon)$.
\end{proof}

\section{Tensorization of product-state errors}\label{sec:product_state_errors}

\begin{lem}\label{lem:product_state_error}
Let $\{\rho_j,\sigma_j\}_{j=1}^N$ be density operators such that $\sigma_j\geq C I$ for some $C>0$. Then
\[
\left\|\bigotimes_{j=1}^N\rho_j-\bigotimes_{j=1}^N\sigma_j\right\|_1
\leq \min\left\{\left(\frac{2}{C}\sum_{j=1}^N\|\rho_j-\sigma_j\|_1^2\right)^{1/2},\sum^N_{j=1}\|\rho_j-\sigma_j\|_1\right\}.
\]
In particular, the left-hand side is at most $\epsilon$ if $\|\rho_j-\sigma_j\|_1\leq \sqrt{C/2}\,\epsilon/\sqrt{N}$ for every $j$.
\end{lem}
\begin{proof}
Let $\rho=\bigotimes_j\rho_j$ and $\sigma=\bigotimes_j\sigma_j$. To get the second term in the right handside, we use the telescoping sum
\[
\rho-\sigma=\sum_{j=1}^N \left(\bigotimes_{k=1}^{j-1}\rho_k\otimes(\rho_j-\sigma_j)\otimes\bigotimes_{k=j+1}^{N}\sigma_k\right),
\]
and then apply the triangle inequality and the submultiplicativity of the trace norm:
\[
\|\rho-\sigma\|_1\leq \sum_{j=1}^N \|\rho_j-\sigma_j\|_1.
\]
By quantum Pinsker's inequality, additivity of relative entropy, and monotonicity of the Petz R\'enyi divergence in its order,
\[
\frac{1}{2}\|\rho-\sigma\|_1^2
\leq D(\rho\|\sigma)
=\sum_jD(\rho_j\|\sigma_j)
\leq\sum_j\log\mathrm{Tr}(\rho_j^2\sigma_j^{-1}).
\]
Since $\sigma_j\geq C I$,
\[
\log\mathrm{Tr}(\rho_j^2\sigma_j^{-1})
=\log\left(1+\mathrm{Tr}\left((\rho_j-\sigma_j)^2\sigma_j^{-1}\right)\right)
\leq C^{-1}\|\rho_j-\sigma_j\|_1^2.
\]
Combining these inequalities proves the result.
\end{proof}

\section{Proof of~\texorpdfstring{\cref{lem:random_bath_ground}}{Lemma~\ref*{lem:random_bath_ground}}}\label{sec:pf_random_bath_ground_lem}
\begin{proof}[Proof of~\cref{lem:random_bath_ground}]
 The error from the infinite time window follows from ~\cite{wang2025lindblad}[Lemma S3]. It is sufficient to consider the adjoint channel,
 \begin{equation}
     \widetilde\Phi_{\alpha, j}^\dagger(O) = \mathcal{U}_{j}^\dagger(T) \circ \widetilde\Psi_{\alpha, j}^\dagger\circ  \mathcal{U}_{j}^\dagger(T) (O),
\end{equation}
where we obtain $\widetilde\Psi_{\alpha,j}^\dagger(O)$ by taking $T\to\infty$ in
\begin{equation}
\begin{aligned}
    \Psi_{\alpha, j}^\dagger(O)
    &= \sum_{n,m\geq 0}(-1)^n(i\alpha)^{n+m}\int  \int_{-T\leq s_n\leq \cdots \leq s_1 \leq T, -T\leq \tau_m\leq \cdots\leq \tau_1 \leq T}\mathrm{d}s_1 \cdots \mathrm{d}s_n \mathrm{d}\tau_1\cdots \mathrm{d}\tau_m \\
        &\sum_{\alpha_i,\beta_j\in\{-1,1\}}  S_{\beta_m,j}^\dag \cdots S_{\beta_1,j}^\dag O S_{\alpha_1,j}\cdots S_{\alpha_n,j}\mathrm{Tr}_{E_j}(B_{\alpha_1,j}\cdots B_{\alpha_n,j}\rho_B(\omega)  B_{\beta_m,j}^\dag \cdots B_{\beta_1,j}^\dag )\\
        & f(s_1)\cdots f(s_n) f(\tau_1)\cdots f(\tau_m)e^{-i(\omega+\lambda_j)\sum_{q=1}^{m}\beta_q\tau_q
+i(\omega+\lambda_j)\sum_{p=1}^{n}\alpha_p s_p}\,. \\
\end{aligned}
\end{equation}
By substituting the corresponding number operator $n_j = d_j^\dagger d_j$, and denote
\begin{equation}
p_{j}(\omega):=\left|\sum_{r=1}^{\infty}(-i\alpha)^{2r-1}\int_{t_1<\cdots<t_{2r-1}}e^{i(\lambda_j+\omega)\sum_{\ell=1}^{2r-1}(-1)^{\ell+1}t_\ell}\prod_{\ell=1}^{2r-1}f(t_\ell)\,\mathrm{d}t_1\cdots\mathrm{d}t_{2r-1}\right|^2 ,
\end{equation}
\begin{itemize}
    \item When $\omega>0$, $\rho_B(\omega) = \ketbra{0}$,
    \begin{equation}
        \begin{aligned}
            \widetilde\Psi_{\alpha, j}^\dagger(n_j)
& = n_j +  \sum_{n,m\geq 0\text{ odd}, n+m\geq2}(-1)^n(i\alpha)^{n+m}\int  \int_{ s_n\leq \cdots \leq s_1,  \tau_m\leq \cdots\leq \tau_1 }\mathrm{d}s_1 \cdots \mathrm{d}s_n \mathrm{d}\tau_1\cdots \mathrm{d}\tau_m\\
            &f(s_1)\cdots f(s_n) f(\tau_1)\cdots f(\tau_m)
         e^{-i(\omega+\lambda_j)\sum_{q=1}^{m}(-1)^{q+1}\tau_q
+i(\omega+\lambda_j)\sum_{p=1}^{n}(-1)^{p+1}s_p}(I-n_j)\\
&= (1-p_{j}(\omega)) n_j + p_j(\omega) I \,. \\
        \end{aligned}
    \end{equation}
    \item When $\omega<0$, $\rho_B(\omega) = \ketbra{1}$,
        \begin{equation}
        \begin{aligned}
            \widetilde\Psi_{\alpha, j}^\dagger(n_j) &= n_j +  \sum_{n,m\geq 0, n+m\geq2, n,m \text{  even}}(-1)^n(i\alpha)^{n+m}\int  \int_{ s_n\leq \cdots \leq s_1 , \tau_m\leq \cdots\leq \tau_1 }\mathrm{d}s_1 \cdots \mathrm{d}s_n \mathrm{d}\tau_1\cdots \mathrm{d}\tau_m\\
            &  f(s_1)\cdots f(s_n) f(\tau_1)\cdots f(\tau_m)  e^{-i(\omega+\lambda_j)\sum_{q=1}^{m}(-1)^{q+1}\tau_q
+i(\omega+\lambda_j)\sum_{p=1}^{n}(-1)^{p+1}s_p} n_j\\
&=(1-p_{j}(\omega)) n_j\,. \\
        \end{aligned}
    \end{equation}
\end{itemize}
Furthermore, $\widetilde{\Phi}_{\alpha,j}^\dagger(n_j) =\widetilde{\Psi}^\dagger_{\alpha,j}(n_j)$. We average over both signs of $\omega$, with $g(\omega)=(2S)^{-1}\mathbf{1}_{[-S,S]}(\omega)$, to obtain
\begin{equation}\label{eq:ground_frequency_average}
    \mathbb{E}_\omega(\widetilde{\Phi}^\dagger_{\alpha,j}(\omega))(n_j) =  \left(1-\int_{\mathbb{R}} g(\omega) p_{j}(\omega)\mathrm{d}\omega\right) n_j + \int_{\mathbb{R}^+} g(\omega)p_j(\omega)\mathrm{d}\omega I
\end{equation}
For $\lambda_j>0$, we apply Lemma S8 in\cite{wang2025lindblad} with $\|f\|_{L^1}=2^{3/4}\pi^{1/4}\sqrt{\sigma}$, to obtain
\begin{equation}
\int_{\mathbb{R}^+} g(\omega)p_j(\omega)\mathrm{d}\omega\leq \int_{\mathbb{R}^+} g(\omega) \left|\sum_{r\geq 1} \alpha^{2r-1} \frac{(2^{3/4}\pi^{1/4}\sqrt{\sigma})^{2r-1}}{(2r-1)!}e^{-\frac{\sigma^2}{2r-1}(\lambda_j + \omega)^2}\right|^2\mathrm{d}\omega
\end{equation}
On $J_j=[-\lambda_j,-\lambda_j+(2\sigma)^{-1}]\subset[-S,0]$, the first-order transition amplitude has magnitude at least $2^{3/4}\pi^{1/4}e^{-1/4}\alpha\sqrt{\sigma}$. We bound the higher odd orders by $\mathcal O((\alpha\sqrt{\sigma})^3)$ and obtain $p_j(\omega)\ge c\alpha^2\sigma$ on $J_j$ for sufficiently small $\alpha\sqrt{\sigma}$ and a constant $c>0$, which gives
\begin{equation}
    \int_{\mathbb{R}^-} g(\omega) p_{j}(\omega)\mathrm{d}\omega\ge C\alpha^2
\end{equation}
Since $\lambda_j+\omega\ge\Delta$ for $\omega>0$ and the amplitude bound is at most one for sufficiently small $\alpha\sqrt{\sigma}$, we obtain
\begin{equation}
\mathbb{E}_\omega(\widetilde{\Phi}^\dagger_{\alpha,j}(\omega))(n_j) \leq (1- C\alpha^2)n_j + \sum_{r\geq 1}\frac{(\mathcal{O}(\alpha \sqrt{\sigma}))^{2r-1}}{(2r-1)!} e^{-\frac{\sigma^2}{2r-1}\Delta^2}I\,.
\end{equation}
For $\lambda_j<0$, we exchange the positive and negative frequency sectors and apply the same argument to $I-n_j$.

\end{proof}

\end{document}